\documentclass[11pt, a4paper]{amsart}

\usepackage{xcolor}
\definecolor{MyBlue}{cmyk}{1,0.13,0,0.63}
\definecolor{MyGreen}{cmyk}{0.91,0,0.88,0.52}
\newcommand{\mylinkcolor}{MyBlue}
\newcommand{\mycitecolor}{MyGreen}
\newcommand{\myurlcolor}{black}

\usepackage{hyperref}
\hypersetup{%
  bookmarksnumbered=true,bookmarksopen=false,%
  plainpages=false,% necessary to prevent duplicate page identifiers
  linktocpage=true,%
  colorlinks=true,breaklinks=true,%
  linkcolor=\mylinkcolor,citecolor=\mycitecolor,urlcolor=\myurlcolor,%
  pdfpagelayout=OneColumn,%
  pageanchor=true,%
}

\usepackage{graphicx}
\usepackage{amsmath, amsthm,  amsfonts, amscd}
\usepackage{amssymb}
\usepackage{url}
\usepackage{fullpage}
\usepackage{mathtools}
\usepackage[all]{xy}
\usepackage{slashed}
\usepackage{multirow}

\newtheorem{thm}{Theorem}[section]
\newtheorem*{thm*}{Theorem}
\newtheorem{cor}[thm]{Corollary}
\newtheorem{lemma}[thm]{Lemma}
\newtheorem{prop}[thm]{Proposition}

\theoremstyle{definition}
\newtheorem{defn}[thm]{Definition}
\newtheorem{assumption}[thm]{Assumption}
\theoremstyle{remark}
\newtheorem{remark}[thm]{Remark}

\newtheorem{example}[thm]{Example}
\newtheorem{remarks}[thm]{Remarks}
\newtheorem{examples}[thm]{Examples}
\newcommand{\R}{\ensuremath{\mathbb{R}}}
\newcommand{\N}{\ensuremath{\mathbb{N}}}
\newcommand{\Z}{\ensuremath{\mathbb{Z}}}

\newcommand{\C}{\ensuremath{\mathbb{C}}}
\newcommand{\T}{\ensuremath{\mathbb{T}}}

\def\calO{\mathcal{O}}
\def\calK{\mathcal{K}}
\def\calB{\mathcal{B}}
\def\calH{\mathcal{H}}

\def\calF{\mathcal{F}}

\def\calN{\mathcal{N}}
\def\calP{\mathcal{P}}

\def\calQ{\mathcal{Q}}

\def\calW{\mathcal{W}}

\def\supp{\mathrm{supp}}
\def\S{\mathbb{S}}

\def\e{\mathrm{e}}

\newcommand{\wh}{\ensuremath{\widehat}}
\newcommand{\wt}{\ensuremath{\widetilde}}
\DeclareMathOperator{\End}{End}

\newcommand{\frakh}{{\mathfrak h}}

\newcommand{\frakC}{{\mathfrak C}}

\newcommand{\one}{{\bf 1}}
\newcommand{\ol}{\overline}
\DeclareMathOperator{\Aut}{Aut}

\theoremstyle{definition}

\DeclareMathOperator{\Index}{Index}

\DeclareMathOperator{\Mult}{Mult}
\DeclareMathOperator{\Ext}{Ext}

\DeclareMathOperator{\Ker}{Ker}

\DeclareMathOperator{\Homeo}{Homeo}

\DeclareMathOperator{\Ran}{Ran}

\newcommand{\Cl}{\C\ell}
\newcommand{\rst}[1]{\ensuremath{{\mathbin\upharpoonright}%
\raise-.5ex\hbox{$#1$}}}

\makeatletter

\newcommand{\Rmnum}[1]{\expandafter\@sl217--242owromancap\romannumeral #1@}
\makeatother

\author{C. Bourne}
\address{Institute of Liberal Arts and Sciences, 
Graduate School of Mathematics, Nagoya University,
Furo-cho, Chikusa-ku, Nagoya, 464-8601, Japan 
\emph{and} RIKEN iTHEMS, 2-1 Hirosawa, Wako, Saitama 351-0198, Japan}
\email{cbourne@nagoya-u.jp}

\date{\today}

\begin{document}

\title{Interfaces   of discrete systems -- spectral and index properties}

\begin{abstract}
We develop a general  mathematical framework to study mixtures of different physical systems 
brought together on a discrete interface. Adapting work by M\u{a}ntoiu et al., we use an operator algebraic framework 
such that the bulk systems at infinity of the mixture are recovered via the spatial 
asymptotics of the  operators on the interface. Fixing an asymptotics and interface algebra, we show how the 
essential spectrum and topological properties can be inferred from the bulk systems at infinity.
By working with Hilbert $C^*$-modules, we can further refine these results with respect to an ambient 
algebra of observables.
\end{abstract}

\maketitle

\tableofcontents

\parindent=0.0in
\parskip=0.06in

\section{Introduction}

The bulk-boundary correspondence of topological phases of matter connects  bulk (boundary-free) properties of a physical system 
 to robust effects localised on its boundary. The phenomenon also provides an interesting connection between materials science 
 with  abstract mathematical concepts such as $K$-theory. In its most basic form, a certain topological invariant, defined on an infinite system 
without boundary, is then connected to a property localised at the co-dimension one boundary of the system restricted to a half space. 
See the monograph~\cite{PSBbook} for a detailed exposition.

The bulk-boundary correspondence has since  been generalised to a much wider variety of geometric settings.
We do not attempt to give a comprehensive review here, though see~\cite{ThiangEdge, LT22, DZ24} for example. 
One may also consider more general defects such as corners and and subsystems of co-dimension two or higher~\cite{Hayashi23, Kubota21, ProdanKK, ProdanHigherOrder}. 
Another possible generalisation is to consider an  interface of two or more systems, see~\cite{KSBV, DZ24, Iwatsuka} for example, where properties of the interface are determined 
by a relative phase of the bulk systems.

Given this plethora of possible systems and boundaries/interfaces/defects one may consider, a general framework to study these systems 
would certainly be desirable. For the case of defects, such a framework has been  provided by Polo Ojito, Prodan and Stoiber~\cite{ProdanHigherOrder}. 
Here our emphasis is on interfaces and we attempt to give a general mathematical formalism for studying mixtures of 
general discrete systems. The formalism certainly has many similarities with and has drawn inspiration from~\cite{ProdanHigherOrder, PPSadiabatic}. 
One important difference is that we  work under the framework of $C^*$-modules.

To give an example as a rough introduction to our approach,  we consider (possibly several) quantum mechanical systems modeled on 
the fixed Hilbert space $\ell^2(\Z^d, \C^N)$. If we wish to consider an interface of dimension $l \leq d$, we take 
the co-dimension $l$ space $\frakh = \ell^2(\Z^{d-l}, \C^N)$, which is placed at every site of the 
interface $\Z^l$ as modeled by Hilbert space $\ell^2(\Z^l, \frakh)$. Clearly 
\[
   \ell^2(\Z^l, \frakh) \cong \ell^2(\Z^l) \otimes \ell^2(\Z^{d-l}, \C^N) \cong \ell^2(\Z^d, \C^N),
\] 
and we recover the total (bulk) Hilbert space. However, the reason for this decomposition is that we can now mix 
different co-dimension $l$ subsystems in $\calB( \frakh)$ in the interface Hilbert space $\ell^2(\Z^l, \frakh)$ via   operators on 
$\ell^2(\Z^l, \frakh)$.

A natural way to consider an interface of $M$ bulk systems on $\Z^l$ is to take a (coarse) 
partition $\{Y_1,\ldots, Y_M\}$ of $\Z^l$ such that for all $j =1,\ldots, M$ each subset $Y_j$ contains a discrete ball 
of arbitrarily large radius. If $\{T_1,\ldots, T_M\}$ are the bulk dynamics,  we can define interface dynamics by 
adding together the compressions of $T_j$ to $Y_j$. 
On the other hand, from the perspective of coarse geometry, 
we do not expect any finite or compact change to an interface to influence its key properties. 
So we expect a clear distinction between the different bulk systems to only become apparent 
in the spatial limit at infinity. As such,  we allow quite general interface 
dynamics, but where   the bulk systems specify the spatial asymptotics of the interface dynamics.

More generally, 
our aim is to consider a wide variety of discrete interfaces. So 
we take $\Gamma$  a countably infinite discrete set and 
$\calH = \ell^2(\Gamma, \frakh)$ as the interface Hilbert space  mixing systems in $\frakh$, which describes 
the    $\Gamma$-defect/co-dimension of a bulk system.  We wish to freely move around 
this interface, so we will assume that $\Gamma$ is a discrete abelian group.
We also  suppose that the $\Gamma$-defect dynamics of all the bulk systems of interest 
can be described by a fixed $C^*$-algebra of observables $B \subset \calB(\frakh)$. 
In the case that $\Gamma = \Z$, $B$ is the algebra of the co-dimension one system, such as the edge  
of a two-dimensional bulk.
The defect operators of interest in $\calB(\frakh)$ may change depending on which bulk system we study, 
but we assume the defect  Hilbert space 
$\frakh$ and $C^*$-algebra $B$ are the same. 
Using the internal algebra 
$B \subset \calB(\frakh)$, we refine the interface Hilbert space $\ell^2(\Gamma, \mathfrak{h})$ to the 
$C^*$-module $\ell^2(\Gamma, B)$ (definitions and basic properties 
are given in the appendix). 

The interface dynamics we consider are adjointable operators $T$ on the $C^*$-module $\ell^2(\Gamma, B)$ 
generated by discrete shifts and multiplication by function in $\ell^\infty(\Gamma, B)$ (such operators are sometimes 
called band-dominated).
The operator $T$ may be  self-adjoint or unitary, which is relevant for Hamiltonian, Floquet and quantum walk 
systems. The asymptotic behaviour of $T$ is controlled so that the various bulk systems can be recovered by taking an appropriate limit at infinity of $T$.

Having introduced this quite general model, our aims for this work are quite modest. Often the property of 
interest in an interface system is closely related to the spectrum the interface dynamics, such as a filling of a bulk spectral gap. 
Computing the spectrum for operators on a  generic interface is quite challenging, though as a first step we can consider the 
essential spectrum. For operators on a Hilbert $C^*$-module, this can be generalised to the essential spectrum 
relative to the ambient $C^*$-algebra $B$. 
To compute this $B$-essential spectrum, an operator algebraic framework has already been 
extensively studied by M\u{a}ntoiu, Georgescu and others~\cite{Mantoiu02, Georgescu1, AMP, Georgescu2, MantoiuGpoid}. 
We adapt this framework to our interface $C^*$-module and show that the $B$-essential spectrum of an interface 
operator is purely determined by the bulk systems at infinity.

A key mathematical element of topological phases is that the bulk system(s) posses a spectral gap. If $\ell^2(\Gamma, B)$ 
represents an interface of a family of (uniformly) gapped bulk systems, we show that an `interface index' can be defined 
via a Fredholm operator $F$ on the interface $C^*$-module $\ell^2(\Gamma, B)$. This index takes values in the $K$-theory 
of the  $\Gamma$-defect algebra $B$. If the system also has free-fermionic symmetries (see~\cite{AMZ} for example), the 
degree of the $K$-theory group will change and the possible symmetries exhaust all 
real and complex $K$-theory groups of $B$. Appealing to general Kasparov theory, we can show that the interface index is 
related to bulk topological properties via a  $K$-theoretic  boundary map, which establishes a weak bulk-interface correspondence.

Even in the case that $B=\C$, computing the value of the interface index in $K_n(\C)$ is quite challenging for complicated 
interfaces. By imposing the strong restriction that the number of bulk systems at infinity are finite and that each system is spatially separated from each other, 
we can decompose the general interface index into a signed sum of refined indices that are related to each bulk system separately. 
In the case of two bulk systems mixed together in a one-dimensional interface/junction/domain wall ($\Gamma = \Z$), we recover a relative index 
$[F] = [F_L] - [F_R]$, where $F_{L/R}$ is a Fredholm operator related to the 
bulk system at $\pm \infty$. More generally, the interface index is trivial when the contributions from each bulk component add to the 
identity. We see this result as a first step towards more comprehensive studies of interfaces/defects whose 
bulk systems may have non-trivial intersection.

% Outline

In Section \ref{sec:DomainWall}, we first consider the relatively simple setting of a one-dimensional interface between two systems, 
often called a domain wall or junction in the physics literature. 
We use this example to give a gentle introduction to our general framework, which is introduced in Section \ref{sec:CStar_mod_interface} 
along with the spectral-theoretic results. A non-exhaustive collection of examples are then considered in Section \ref{sec:B-valued_examples} 
to give a flavour of what interfaces are possible.
In Section \ref{sec:Interface_index} we consider the $K$-theoretic properties of this interface, where the interface index is defined as well as 
its basic decomposition properties. Appendix \ref{sec:Kasparov_review} and \ref{sec:extensions} provide a basic introduction to the 
relevant background on Kasparov theory and the extension semigroup of $C^*$-algebras.

\subsubsection*{Acknowledgements}
This work is supported by a JSPS Kakenhi Grant-in-Aid (24K06756) and has greatly benefited from discussions 
with B. Mesland,  D. Polo Ojito, E. Prodan, S. Richard and T. Stoiber.

\section{Warm-up: domain wall ($\Gamma = \Z$)} \label{sec:DomainWall}

Before looking at the  general setting, let us start with a one-dimensional example. 
We consider two `bulk' systems, whose dynamics on a co-dimension one subsystem or boundary 
is described by a unital and separable $C^*$-algebra of observables 
$B \subset \calB(\mathfrak{h})$. 
To consider a one-dimensional interface of these systems, we take the lattice $\Z$ and 
Hilbert space $\ell^2(\Z, \mathfrak{h}) \cong \ell^2(\Z) \otimes \mathfrak{h}$. Note 
that although the interface $\Z$ is one-dimensional, $\mathfrak{h}$ may be describing a system of 
arbitrary dimension.

The interface dynamics $T \in \calB\big( \ell^2(\Z, \frakh) \big)$ should be such that the limits at $\pm \infty$ 
 recover operators that act on separate bulk systems that can be considered without reference to the 
interface. That is, our bulk systems $T_L, T_R$ on $\ell^2(\Z, \frakh)$ (with $\ell^2(\Z, \frakh)$ now considered as a bulk Hilbert space) 
arise as a spatial limit at infinity of the 
interface operator $T$.\footnote{One may instead be interested 
in unbounded operators  affiliated to the relevant algebra of observables. Because our setting is discrete systems, the 
operators we are interested in studying (discrete Schr\"{o}dinger operators, quantum walk and Floquet operators) are bounded. 
So we will work in this simpler setting.}

Of course, we may consider a model where $T=T_L$ for $x< 0$ and $T=T_R$ for $x\geq 0$ and there is a 
clear boundary between the two systems. But if we are considering systems at a quantum scale, we can not expect 
that such a clean boundary can be realised in general. So instead we consider a mixing of the two systems on the interface $\Z$, 
but with asymptotics such that we can distinguish the two bulk systems at $\pm \infty$.

Using the co-dimension one observable algebra $B \subset \calB(\mathfrak{h})$, we can extract more refined information. Namely, 
rather than the Hilbert space $\ell^2(\Z, \mathfrak{h})$, we consider the space $\ell^2(\Z, B)$, which is 
a right-$B$ Hilbert $C^*$-module (see Appendix \ref{sec:Kasparov_review} for basic definitions and properties). We can think of 
$\ell^2(\Z, B)$ as a one-dimensional interface of the fixed algebra $B$ that is describing the co-dimension one dynamics. 
In the language of $C^*$-modules, 
the dynamics of this interface are modeled by adjointable operators $T \in \End_B( \ell^2(\Z, B) ) \cong \Mult\big( \calK(\ell^2(\Z)) \otimes B \big)$, 
where $\Mult(A)$ is the multiplier algebra of $A$ (the largest possibly unitisation).

Our key hypothesis is that this mixing operator $T \in \End_B( \ell^2(\Z, B) ) $  can be described as some combination of discrete shifts 
and pointwise multiplication 
\[
   T = \sum_{n\in \Z}^{\text{finite}} S^n f_n, \qquad  (S\psi)(x) = \psi(x-1), \qquad f_n \in \ell^\infty(\Z, B), \qquad (f_n\psi)(x) = f_n(x) \psi(x)
\]
with $x \in \Z$ and $\psi \in \ell^2(\Z, B)$.
Operators that can be norm-approximated by such sums  are called band-dominated in the Hilbert space 
literature (see~\cite{RSS98} for example).

Understanding $\calH = \ell^2(\Z, B)$ as a domain wall between two separate systems, we  further 
impose conditions on the multiplication operators $f \in  \ell^\infty(\Z, B)$. Namely, if $T$ is modeling the mixed system 
of $T_L$ and $T_R$, we expect the limits at $-\infty$ and $+\infty$ to recover these operators. Hence, we restrict our 
attention to multiplication operators
\[
  f \in A := \big\{ f \in  \ell^\infty( \Z, B) \, \big| \, \lim_{x \to \pm \infty} f(x) =: f(\pm \infty) \in B \text{ exists} \big\}. 
\]
Such an algebra is invariant under the natural action of $\Z$ on $\ell^\infty( \Z, B)$ by translations, 
$\alpha: \Z \to \Aut\big( \ell^\infty( \Z, B) \big)$, $\big(\alpha_k (f) \big)(x) = f(x-k)$ for any $f \in A$ and $k,x \in \Z$.
We also impose the following technical assumption that 
the range of  $f \in A$ for all such $f\in A$ recovers the ambient $C^*$-algebra $B$ of the co-dimension one system. 
In the case that $B = M_N(\C)$, this assumption ensures that the image of $f(x)$ in $B$  
does not contain a block matrix of $0$s. That is, we take functions that see all degrees of freedom in $B$.
With these assumptions, the Gelfand--Naimark Theorem implies that we can explicitly characterise the algebra of multiplication operators  
as $A \cong C( \ol{\Z}, B)$, where  $\ol{\Z} = \Z \cup \{-\infty, + \infty\}$ is the closure of $\Z$ in $[-\infty, \infty]$,  the two-point compactification 
of $\R$.

Of course, the interface dynamics $T$ also involves the discrete shift operators $\{S^k\}_{k \in \Z}$ on 
$\ell^2(\Z, B)$. Conjugating functions $f \in A$ by these shift operators implements the translation action 
$S^k f S^{-k} = \alpha_k(f)$ for all $k \in \Z$. This means that the interface operator $T$ can be considered as an element 
of the crossed product $C^*$-algebra 
\[
   T \in A \rtimes \Z = C( \ol{\Z}, B) \rtimes \Z \subset \End_B(\ell^2(\Z, B)) \cong  \Mult\big( \calK(\ell^2(\Z)) \otimes B \big),
\]
see the texts~\cite{Pedersen, Wil07} for further details on crossed product $C^*$-algebras.
Because $\Z$ is amenable, we do not distinguish between the full and reduced crossed product.
The algebra $A=  C(\ol{\Z}, B)$ and its crossed product by $\Z$  fit into the short exact sequences,
\begin{align} \label{eq:DW_extension}
   0 \to C_0( \Z, B) \to &C( \ol{\Z} , B) \to C( \{-\infty, +\infty\}, B) \to 0, \nonumber \\
   0 \to C_0( \Z, B)\rtimes \Z \to &C( \ol{\Z} , B)\rtimes \Z \to C( \{-\infty, +\infty\}, B)\rtimes \Z \to 0.
\end{align}
Furthermore, we can simplify
\[
   C( \{-\infty, +\infty\}, B) \cong B^{\oplus 2}, \quad 
   C( \{-\infty, +\infty\}, B)\rtimes \Z \cong \big( B \rtimes \Z\big)^{\oplus 2}, \quad 
   C_0( \Z, B)\rtimes \Z \cong \calK\big(\ell^2(\Z) \big) \otimes B.
\]
In this example, we have taken the action of $\Z$ on $B$ to be trivial, so  $B \rtimes \Z \cong C(\T, B)$, though we will generally 
work in the case that $\Z$ also acts non-trivially on $B$. 
Note that $C_0( \Z, B)\rtimes \Z \cong \calK\big(\ell^2(\Z) \big) \otimes B$ also holds for non-trivial 
actions by Takai duality~\cite[Theorem 7.9.3]{Pedersen}.

The quotient $C( \ol{\Z} , B) \to C( \{-\infty, +\infty\}, B) \cong B^{\oplus 2}$ is precisely the evaluation of 
$f \in C(\ol{\Z}, B)$ at the points $\pm \infty$. Hence the quotient localises the mixed system at the spatial 
location of the bulk systems at infinity.
However, the algebra $B \subset \calB(\frakh)$ is a co-dimension one algebra of the bulk. So  to 
fully recover the bulk systems, we also need to incorporate the $\Z$-dynamics on $B$ 
 via an action $\beta: \Z \to \Aut(B)$. 
Thus, for $f \in A=C( \ol{\Z} , B)$, our action $\alpha: \Z \to \Aut(A)$ is given by $\big(\alpha_k(f)\big)(x) = \beta_k\big( f(x-k) \big)$ 
for $k, x \in \Z$ and $\big( \alpha_k(f) \big)( \pm \infty) = \beta_k\big( f(\pm \infty) \big)$. 
Localising the system at infinity, the crossed product $ C( \{-\infty, +\infty\}, B)\rtimes \Z \cong \big( B \rtimes \Z\big)^{\oplus 2}$ 
recovers the two  bulk systems. In particular, at $\pm \infty$ the separate bulk system $B \rtimes \Z$ can be faithfully represented 
on $\ell^2(\Z, \frakh)$, which can now be considered as the bulk Hilbert space. We remark that 
while the details are somewhat involved, the map from the interface to a bulk system is quite simple, namely 
\[
   q_{L/R}(T) = \Big( \sum_n S^n f_n \Big) = \sum_n S^n f_n(\pm \infty) = T_{L/R} \in B\rtimes \Z.
\]

To summarise our discussion thus far, we model the domain wall system by the $C^*$-algebra 
$C( \ol{\Z} , B)\rtimes \Z$. The quotient of this algebra by the ideal $C_0( \Z, B)\rtimes \Z \cong \calK\big(\ell^2(\Z) \big) \otimes B$ 
gives the two distinct bulk systems localised at $\pm \infty$.

With this long introduction done, we can now explain our aims for such a model.
\begin{enumerate}
  \item[(i)] Understand spectral properties of an interface operator  $T \in C(\ol{\Z} , B) \rtimes \Z$, in terms of 
  the bulk operators at infinity $(T_L, T_R) \in \big( B \rtimes \Z\big)^{\oplus 2}$. 
  \item[(ii)] If we suppose that the bulk systems at infinity are gapped, so there are invertible operators 
  $F_L$ and $F_R$ in $B\rtimes \Z$, determine any topological information about a lift 
  $F \in C(\ol{\Z} , B) \rtimes \Z$ of $(F_L, F_R) \in ( B \rtimes \Z )^{\oplus 2}$.
\end{enumerate}

In this  setting of a domain wall, both of these questions can be answered relatively easily.

\begin{enumerate}
  \item[(i)] For any operator $T \in \End_B(\ell^2(\Z, B) ) \cong \Mult\big( \calK(\ell^2(\Z)) \otimes B \big)$, we define the $B$-essential spectrum
  \[
     \sigma^{B}_{\mathrm{ess}}(T) = \sigma(q(T)), \quad 
     q: \End_B( \ell^2(\Z, B) ) \to \End_B(\ell^2(\Z, B))/ \mathbb{K}_B(\ell^2(\Z, B) ) ,
  \]
  where $ \mathbb{K}_B(\ell^2(\Z, B) ) \cong \calK( \ell^2(\Z) ) \otimes B$ are the compact operators on the 
  interface $C^*$-module $\ell^2(\Z, B)$. 
  When $B \subset \calK(\frakh)$, the compact operators on $\frakh$, we recover the usual notion of the essential spectrum of an operator on a Hilbert space.
  For any $T \in C(\ol{\Z}, B) \rtimes \Z$, the short exact sequence from Equation  \eqref{eq:DW_extension} gives us that 
  \[
      \sigma_\mathrm{ess}^B(T) = \sigma( T_L) \cup \sigma(T_R), \qquad 
      T_L, T_R \in B\rtimes \Z.
  \]
  Hence the $B$-essential spectrum of $T \in C(\ol{\Z}, B) \rtimes \Z$ can be entirely understood by the 
  asymptotic bulk limits $(T_L, T_R)$.
  \item[(ii)] Our gapped bulk assumption says that $q(F) = (F_L, F_R)$ is invertible. 
  Hence the operator $F \in \End_B(\ell^2(\Z, B) )$ is invertible modulo the compacts $\mathbb{K}_B(\ell^2(\Z, B) )$ 
  and so is a Fredholm operator on $\ell^2(\Z, B)$ (see Section \ref{sec:KK} for further details). 
  In such a setting, we can define a $K$-theoretic index 
  $[F] \in KK(\C, B) \cong K_0(B)$. We can also understand $[F]$ as the image of boundary map 
  \[
     K_1\big( ( B \rtimes \Z)^{\oplus 2} \big) \ni [(F_L, F_R)] \xmapsto{\delta} [F] \in K_0(B)
  \]  
 that is associated to the 
  short exact sequence 
  \[
     0 \to \calK\big(\ell^2(\Z) \big) \otimes B \to C( \ol{\Z}, B) \rtimes \Z \to  (B\rtimes \Z)^{\oplus 2} \to 0.
  \]
  Because the two bulk systems at $\pm \infty$ are spatially separated, the invertible elements $F_L$ and $F_R$ 
  can be represented in separate exact sequences,
 \[
     0 \to \calK\big(\ell^2(\Z) \big) \otimes B \to C_0( \Z\cup \{\pm \infty\}, B) \rtimes \Z 
       \xrightarrow{\mathrm{ev}_{\pm \infty}}  B\rtimes \Z \to 0.
  \]
  These one-sided short exact sequences can also be understood (up to a sign) as the Toeplitz extension 
  for the crossed product $B\rtimes \Z$. One then finds that (see~\cite[Lemma 5.4]{B22})
  \begin{equation} \label{eq:domwall_bdry_map_decomp}
     \delta \big[ (F_L, F_R) \big] = \delta_\mathrm{Toep} [F_L] - \delta_\mathrm{Toep} [F_R] \in K_0(B ),
  \end{equation}
  where $\delta_\mathrm{Toep}$ is the boundary map of the one-sided extension and where $[x]-[y] = [x] \oplus [y]^{-1}$ 
  in the abelian group $K_0(B )$. The reason for the minus sign/inverse is due to the orientation change 
  that comes with considering the limit at $+\infty$ or $-\infty$. Equation \eqref{eq:domwall_bdry_map_decomp} 
  shows that we can decompose the $K$-theoretic index $[F] = [\tilde{F}_L] - [\tilde{F}_R]$ with $\tilde{F}_{L/R}$ 
  generalised Fredholm operators in $C_0( \Z\cup \{\pm \infty\}, B) \rtimes \Z \subset \End_B(\ell^2(\Z, B))$.
  
  The relation $[F]=[\tilde{F}_L] - [\tilde{F}_R]$ also tells us that if the 
  separate bulk gapped systems are topologically equivalent,  $[F_L]=[F_R] \in K_1(B\rtimes \Z)$, 
  then  the interface index $[F] \in K_0(B)$ will be trivial. Put another way, if the 
  interface index is non-trivial, then our separated bulk systems are topologically distinct. 
  Such a result is a weak version of a $K$-theoretic bulk-boundary correspondence for a domain wall interface. 

\end{enumerate}

We lastly remark that, until this point, we have taken $B$ to be an arbitrary unital and separable $C^*$-algebra. 
In many cases of interest $B$ is itself given by a crossed product $C(X) \rtimes \Z^m$ 
with $X$ some compact Hausdorff space of configurations equipped with an ergodic probability measure. In such a setting, computing 
the spectrum of an element $T_{L/R} \in B \rtimes \Z^m$ is much easier than computing the spectrum of a general operator 
$T \in \End_B(\ell^2(\Z, B))$. So the $B$-essential spectrum 
\[
 \sigma^{B}_{\mathrm{ess}}(T) = \sigma( T_L) \cup \sigma(T_R)
\]
is a much more tractable quantity in many examples of interest. This is particularly true when $B$ models a 
periodic system and we can use the Fourier transform to compute $\sigma(T_{L/R})$.

\section{Discrete interfaces via $C^*$-modules -- framework and spectral properties} \label{sec:CStar_mod_interface}

Here we present an abstract but flexible tool to analyse a 
discrete interface of a quantum system described via a unital $C^*$-algebra. 
We let $\Gamma$ be a discrete, countable and abelian group. In many of the examples of interest to us, 
$\Gamma = \Z^l$ for some $l \in \N$, which describes an interface or defect with co-dimension $l$.
Our approach is via Hilbert $C^*$-modules, where a basic overview can be found in Appendix \ref{sec:Kasparov_review}.

\subsection{Setting and assumptions}

We fix a unital and separable $C^*$-algebra $B$ possibly with internal dynamics $\beta: \Gamma \to \Aut(B)$. 
For technical reasons, we will restrict to the case that $B$ is a nuclear $C^*$-algebra, though this can be weakened 
to assuming that $B$ is exact (the minimal tensor product with $B$ preserves short exact sequences). 
In typical examples arising from condensed matter physics, $B$ may be an AF algebra, 
the $C^*$-algebra of an amenable transformation groupoid, or the $C^*$-algebra of an 
\'{e}tale and amenable groupoid, which are nuclear.

For an interface $\Gamma$, we think of  $B$ as the algebra of observables of a $\Gamma$-defect subsystem 
of the bulk. If $X$ is a compact Hausdorff 
space and $C(X) \rtimes \Z^d$ 
is the bulk algebra, then for  $\Gamma = \Z^l$ with $l \leq d$, we take $B = C(X) \rtimes \Z^{d-l}$ as the algebra of  the co-dimension $l$ subsystem.

Returning to the general picture,  a   discrete interface of a system described by $B$ distributed on the lattice $\Gamma$ 
can be modeled by the Hilbert $C^*$-module 
$\ell^2(\Gamma, B)$. The dynamics of this interface are described by adjointable operators 
$T \in \End_B\big( \ell^2(\Gamma, B) \big)\cong \Mult\big( \calK(\ell^2(\Gamma)) \otimes B\big)$. 

\begin{remark}
 
By working with $\ell^2(\Gamma, B)$, we are assuming that the mixing and interaction between each of the lattice points
$\big\{\{x\} \otimes B\big\}_{x \in \Gamma}$ is quite mild. Indeed, any element $\psi \in \ell^2(\Gamma, B)$ is such that 
\[
    \sum_{x\in \Gamma} \psi(x)^* \psi(x) \  \text{converges in } B.
\] 
 If there are long-range dynamics and interactions, then it is more
suitable  to work with the $C^*$-algebra $\bigotimes_\Gamma B$ with dynamics given by 
elements in $\Aut\big( \bigotimes_\Gamma B \big)$. We leave an investigation of this more challenging setting to another place.
\end{remark}

If we are given finitely many bulk systems to mix together in an interface, there are many different ways we can specify a 
dynamics. For example, we may  take an appropriate partition of $\Gamma$ and sum together the restriction of 
each bulk operator to  one piece of the partition. 
Many of these choices are quite arbitrary and one does not expect  any fundamental property of the interface to 
depend on such choices. We take the perspective that, in general, one can not distinguish the 
individual bulk systems within any finite region of a generic interface. Instead the bulk data is recovered in an appropriate 
limit at infinity. As such, we will allow for quite general dynamics on the interface and employ operator algebraic 
methods to recover the bulk information at infinity.

Let us now  specify the types of dynamics of the interface $\ell^2(\Gamma, B)$ of interest to us:
\begin{enumerate}
  \item Pointwise multiplication by a function, $(f\cdot \psi)(x) = f(x)\psi(x)$ with $\psi \in \ell^2(\Gamma, B)$ and 
  \[
     f \in \ell^\infty(\Gamma, B) = \big\{ f : \Gamma \to B \, \mid \, f \text{ bounded} \big\} \cong \Mult\big(C_0(\Gamma, B) \big) \cong C(\beta \Gamma, B),
  \]
  where $\beta\Gamma$ denotes the Stone--\v{C}ech compactification of $\Gamma$,
  \item Discrete lattice shifts twisted by a (possibly trivial) action of $\Gamma$ on the internal algebra $B$, 
  \[
     \big\{S_g\big\}_{g \in \Gamma}, \quad (S_g \psi)(x) = \beta_g \big(\psi(g^{-1}x)\big), \quad \psi \in \ell^2(\Gamma, B), \,\, x \in \Gamma,  \quad 
     \beta : \Gamma \to \Aut(B).
  \]
\end{enumerate}

Operators $T \in \End_B( \ell^2(\Gamma, B) )$ that are generated from shifts and bounded multiplication are sometimes 
called band-dominated operators in the Hilbert space setting~\cite{RSS98}. 
The isolated bulk systems come from a spatial asymptotic limit of the operators in $\End_B(\ell^2(\Gamma, B))$. 
We will therefore consider a $C^*$-subalgebra of band-dominated operators, where the spatial asymptotics are 
 tuned by specifying  the types of pointwise multiplication 
we allow. Let us therefore explicitly state our operating assumptions.

\begin{assumption} \label{assump:B-valued_functions}
We consider a $C^*$-subalgebra $A \subset \ell^\infty(\Gamma, B)$ such that
\begin{enumerate}
   \item The action  $\alpha: \Gamma \to \Aut\big(\ell^\infty(\Gamma, B)\big)$ by translations restricts to an 
   action of $A$,
   \[
      \alpha: \Gamma \to \Aut(A), \qquad [\alpha_g( f )](x) = \beta_g \big( f(g^{-1}x) \big), \quad f\in A, \quad x, g \in \Gamma, 
   \]
   \item The algebra $A$ contains $C_0(\Gamma, B)$ as a subalgebra and therefore as a closed two-sided ideal. 
   \item The set $\big\{ \Ran(f) \, \big| \, f \in A \big\} = B$.
\end{enumerate}
\end{assumption}

The first condition ensures that we are free to move around the interface. The functions vanishing at infinity are precisely those 
that will not interact with the bulk systems at infinity. Hence we wish to include $C_0$-functions as part of  the interface (non-bulk) dynamics. 
We also note that the ideal $C_0(\Gamma, B) \subset A$  is essential, if $f\in A$ and 
$f \cdot C_0(\Gamma, B) = 0$, then $f=0$. 
Therefore there is an injective $\ast$-homomorphism $\nu: A \to   \Mult\big( C_0(\Gamma, B) \big)$ and we 
can naturally consider  
$A \subset \Mult\big( C_0(\Gamma, B) \big) \cong C(\beta \Gamma, B)$, 
where $\beta \Gamma$ is the Stone--\v{C}ech compactification of $\Gamma$.

Considering $ A \subset \ell^\infty(\Gamma, B) \cong C(\beta \Gamma, B)$, the last condition ensures that we see all of the 
internal algebra $B$ and not a subspace. For example, if $B= M_N(\C)$ for  $N \geq 2$, 
we are not restricted to functions of the form 
\[
    f(x) = \begin{pmatrix} \tilde{f}(x) & \mathbf{0}_{N-k, k} \\ \mathbf{0}_{N-k, N} & \mathbf{0}_{N-k, N-k} \end{pmatrix}, \quad k < N.
\]

An immediate consequence Assumption \ref{assump:B-valued_functions} and the Gelfand--Naimark Theorem is the following.
\begin{lemma}
If $A \subset \ell^\infty(\Gamma, B)$ satisfies Assumption \ref{assump:B-valued_functions}, then $A \cong C_0(\Omega, B)$ for some  
$\Gamma$-invariant, locally compact and Hausdorff space $\Omega$ such that $\Gamma  \subset \Omega \subset  \beta\Gamma$. 
If $A$ is unital, then $\Omega$ is a compactification of $\Gamma$.
\end{lemma}

\begin{remark}
An algebra $A$ satisfying  Assumption \ref{assump:B-valued_functions} may be unital or non-unital. In what follows, we will present our results 
under the assumption that $A$ is non-unital as more care is required. The results for the case that $A$ is unital (which includes many of the examples of interest) 
can then be easily obtained by some simple adjustments.
\end{remark}

By the action of $\Gamma$ on $A$, there is an induced dynamical system $\Gamma \to \Homeo(\Omega)$ such that 
$[\alpha_g(f)](\omega) = \beta_g\big(f( g^{-1} \cdot \omega)\big)$ for $g \in \Gamma$, $\omega \in \Omega$ and $f \in C_0(\Omega, B)$.

Because $C_0(\Gamma, B) \subset A \cong C_0(\Omega, B)$ is a closed two-sided ideal, there is a short exact sequence
\begin{equation} \label{eq:B-valued_function_extension}
   0 \to C_0(\Gamma, B) \to C_0(\Omega, B) \to C_0( \Omega \setminus \Gamma, B) \to 0,
\end{equation}
where $C_0( \Omega \setminus \Gamma, B) \cong C_0(\Omega, B) / C_0(\Gamma, B)$. 
In what follows we denote $\partial \Omega := \Omega \setminus \Gamma$. 
The quotient $C_0( \partial \Omega, B)$ gives the spatial information of the system at infinity. The various bulk systems 
are spatially located in  certain $\Gamma$-invariant subsets of $\partial \Omega$ (see Section \ref{subsec:bulk_and_quasiorbits} below).

Returning to operators $T \in \End_B\big( \ell^2(\Gamma, B) \big)$ acting on the discrete interface, we can consider the restricted sum 
\[
    T = \sum_{g \in \Gamma}^{\text{finite}} S_g f_g, \quad f_g \in C_0(\Omega, B),
\]
where by assumption $S_g f = \alpha_g (f) S_g$ for any $f \in C_0(\Omega, B)$. 
Such finite sums will give a faithful representation of the 
$\ast$-algebra $C_c(\Gamma, B)$ on $\ell^2(\Gamma, B)$. Because $\Gamma$ is amenable, the completion of 
$C_c(\Gamma, B)$ in $\End_B\big( \ell^2(\Gamma, B) \big)$ gives a 
faithful representation of the crossed product algebra $C_0(\Omega, B) \rtimes \Gamma$. 
Hence, choosing a  spatial asymptotics as encoded by an algebra  $A$ satisfying Assumption \ref{assump:B-valued_functions}, 
our basic model for the interface dynamics is determined by the crossed product $C^*$-algebra 
$C_0(\Omega , B)\rtimes \Gamma$ with $\Omega \subset \beta \Gamma$ a $\Gamma$-invariant subset.

Because the short 
exact sequence of Equation \eqref{eq:B-valued_function_extension} is $\Gamma$-equivariant, we can also consider the 
extension
\[
  0 \to C_0(\Gamma, B)\rtimes \Gamma \to C_0(\Omega, B) \rtimes \Gamma \to C_0( \partial \Omega  , B) \rtimes \Gamma \to 0.
\]
By Takai duality~\cite[Theorem 7.9.3]{Pedersen}, we can further identify $C_0(\Gamma, B)\rtimes \Gamma \cong \calK\big(\ell^2(\Gamma) \big) \otimes B$ 
and have the short exact sequence,
\begin{equation} \label{eq:B-valued_crossed-prod_extension}
   0 \to \calK\big(\ell^2(\Gamma) \big) \otimes B \to C_0(\Omega, B) \rtimes \Gamma \to C_0( \partial \Omega, B) \rtimes \Gamma \to 0.
\end{equation}
The ideal $C_0(\Gamma, B) \rtimes \Gamma \subset \End_B(\ell^2(\Gamma, B))$  acts on the interface $C^*$-module $\ell^2(\Gamma, B)$ and 
precisely gives the compact operators $\mathbb{K}_B(\ell^2(\Gamma, B)) \cong \calK\big(\ell^2(\Gamma) \big) \otimes B$.

\begin{remark} \label{rk:Stone-vonNeumann_rk}
If $G$ is a locally compact group that acts trivially on $B$, then 
\[
   C_0(G, B) \rtimes G \cong \big( C_0(G) \otimes B \big) \rtimes_{\tau \otimes \mathrm{Id} } G \cong \big( C_0(G) \rtimes_\tau G \big) \otimes B 
   \cong \calK[ L^2(G) ] \otimes B
\]
by the Stone--von Neumann Theorem~\cite[Theorem 4.24]{Wil07} and we do not need to invoke Takai duality.
\end{remark}

\begin{remark}
 We will mostly think of $\ell^2(\Gamma, B)$ as describing an interface of the $\Gamma$-defect algebra 
 $B$, though the mathematical formalism can be applied without reference to interfaces of bulk systems. This is particularly relevant for 
 quantum walks, where we can regard a unitary  $U \in M_n \big(C_0(\Omega, B) \rtimes \Gamma \big)$ as describing a particular 
 discrete time step with spatial asymptotics specified by the algebra $C_0(\Omega, B) \subset \ell^\infty(\Gamma, B)$.
\end{remark}

\subsection{Bulk systems and quasi-orbits} \label{subsec:bulk_and_quasiorbits}

Let us consider the quotient $C_0( \partial \Omega, B)$ and the crossed product 
$C_0( \partial \Omega, B) \rtimes \Gamma$. The topological space $\partial \Omega = \Omega \setminus \Gamma$ specifies the 
spatial location of the operators at infinity, which should not be affected by the dynamics that come from 
the mixing of different systems in the interface. As such, we expect the various bulk systems to be spatially supported 
in $\partial \Omega$. So for a function $f \in C_0( \partial \Omega, B)$ and $\omega \in  \partial \Omega$, 
the element $f(\omega) \in B$ gives an operator in the internal system $B$ that is sufficiently far away from the interface so 
as to not be affected by its dynamics. 

On the other hand, given $f \in C(\partial \Omega, B)$, we have no canonical way to distinguish an 
element $f(\omega) \in B$ with the finite shift 
$[\alpha_g( f)](\omega) = \beta_g\big( f(g^{-1}\cdot \omega) \big)$ for any $g \in\Gamma $ and $\omega \in \Omega\setminus \Gamma$. 
We therefore consider the elements $f(\omega)$ and $[\alpha_g( f)](\omega) \in B$ as coming from the \emph{same} 
bulk system. Indeed, because we think of $B$ as a $\Gamma$-defect algebra, the algebra $C_0(\partial \Omega, B)$ is not 
enough to recover the bulk systems and the action by $\Gamma$ needs to be taken into account.
Given any point at infinity, $\omega \in \partial \Omega$, the bulk system at $\omega$
 is modeled by $C_0( \Xi_\omega, B) \rtimes \Gamma$, where 
$\Xi_\omega = \ol{\big\{ g^{-1}\cdot \omega \mid g \in \Gamma  \big\}}$ is the closure of the 
orbit of $\omega$ (with closure in $\partial \Omega$). In analogy to the work of M$\breve{\text{a}}$ntoiu~\cite{Mantoiu02}, we 
call the space $\Xi_\omega$ the quasi-orbit of the point $\omega$. 

More generally, we may wish to consider generic $\Gamma$-invariant subsets $Z \subset \Omega\setminus\Gamma$, which may contain several 
distinct bulk systems (several quasi-orbits). Distinct quasi-orbits $\Xi_\omega$ and $\Xi_{\omega'}$ need not be disjoint and 
will depend on the asymptotics under consideration. Quasi-orbits with non-trivial intersections are  quite common in systems with 
defects such as corners and hinges~\cite{ProdanHigherOrder, PPSadiabatic} (see also Section \ref{sec:B-valued_examples} below).

We also remark that, in the same way that $\Omega$ is a $\Gamma$-invariant subset of $\beta \Gamma$, the quotient $\partial \Omega$ is a $\Gamma$-invariant 
subset of the compact space $\beta \Gamma \setminus \Gamma$.

\subsection{$B$-essential spectrum}

A generic operator on the interface $C^*$-module is an element
$T \in \End_B\big(\ell^2(\Gamma, B) \big) \cong \Mult\big( \calK(\ell^2(\Gamma)) \otimes B\big)$ with 
$\calK(\ell^2(\Gamma)) \otimes B \cong \mathbb{K}_B\big( \ell^2(\Gamma, B) \big)$ the compact endomorphisms 
on $\ell^2(\Gamma, B)$. Recalling the essential spectrum on Hilbert spaces and generalised Calkin algebra 
$\calQ_B(\ell^2(\Gamma, B) ) = \End_B( \ell^2(\Gamma, B) ) / \mathbb{K}_B(\ell^2(\Gamma, B) )$ from Appendix \ref{sec:Kasparov_review},
 we define the following.

\begin{defn}
Let $T \in \End_B\big(\ell^2(\Gamma, B) \big)$. The $B$-essential spectrum of $T$ is the set 
\[
  \sigma^{B}_{\mathrm{ess}}( T) = \sigma\big( q(T) \big), \qquad 
  q: \End_B\big(\ell^2(\Gamma, B) \big) \to \calQ_B \big(\ell^2(\Gamma, B) \big) 
  \cong \calQ\big( \calK(\ell^2(\Gamma)) \otimes B\big).
\]
\end{defn}

When $B\subset \calK$ for $\calK$ the compact operators on some Hilbert space, then 
$\sigma^{B}_{\mathrm{ess}}( T) = \sigma_{\mathrm{ess}}( T)$.

Let us  consider $A \cong C_0(\Omega, B)$ 
satisfying Assumption \ref{assump:B-valued_functions} and an element $T \in A \rtimes \Gamma$ with the aim of 
computing $\sigma^{B}_{\mathrm{ess}}( T)$. 
It is an immediate consequence of Equation \eqref{eq:B-valued_crossed-prod_extension} that 
$\sigma^{B}_{\mathrm{ess}}( T)$ is given by the spectrum of $q(T) \in C_0(\partial \Omega, B) \rtimes \Gamma$. 
On the other hand, in the same way that there is not a faithful representation of 
$\calQ(\calH)$ on $\calH$,  the algebra $C_0(\partial \Omega, B) \rtimes \Gamma$ might not have a faithful 
representation on the  $C^*$-module $\ell^2(\Gamma, B)$. So it may be difficult to 
work with the abstract element $q(T) \in C_0(\partial \Omega, B) \rtimes \Gamma$ directly. 
We can deal with this problem by 
instead considering the various bulk systems represented by quasi-orbits in $\partial\Omega$.

Let us consider a decomposition of $\partial\Omega$ in quasi-orbits. 
By choosing sufficiently many $\omega_j\in \Omega\setminus \Gamma$, 
one gets a covering of $\partial \Omega$,  
\begin{equation}\label{eq:covering}
  \Omega\setminus \Gamma =  \partial \Omega =  \bigcup_{j \in J} \Xi_{\omega_j}.
\end{equation}
For simplicity, we shall write $\Xi_j$ for $\Xi_{\omega_j}$.
Based on this decomposition, one directly infers the following result.

\begin{lemma}\label{lem:quotient}
Let $A$ satisfy Assumption \ref{assump:B-valued_functions}
and consider a covering of $\partial \Omega$ by quasi-orbits satisfying Equation \eqref{eq:covering}.
Then there exists two injective $*$-homomorphisms,
\[
 A/ C_0(\Gamma, B) \to \prod_{j\in J} C_0( \Xi_j, B)
\]
and
\[
 A \rtimes \Gamma \big/  C_0(\Gamma, B)\rtimes \Gamma  \to \prod_{j\in J} C_0(\Xi_j, B) \rtimes \Gamma .
\]
\end{lemma}

Note that the map $A/ C_0(\Gamma, B) \to \prod_{j\in J} C_0( \Xi_j, B)$ is rarely surjective, since generally there are compatibility conditions 
between the various restrictions on $\Xi_j$ of any function defined on $\Omega$.
Since the set $\Xi_j$ is a closed $\Gamma$-invariant subset of $\Omega$, the algebra
\begin{equation}\label{eq_Aj}
 A_j:=\big\{  f \in C_0(\Omega, B)  \, \big| \, f|_{\Xi_j}=0\big\}
\end{equation}
corresponds to a $\Gamma$-invariant ideal of $A$. By construction, 
 the equalities 
\[
\bigcap_{j\in J} A_j=C_0(\Gamma, B), \qquad \quad 
 \bigcap_{j\in J} A_j\rtimes \Gamma \cong \calK\big( \ell^2(\Gamma)\big) \otimes B
\]
hold. Using $A_j$, we can also consider the  short exact sequences
\[
0 \to A_j \to A \to C_0(\Xi_j, B) \to 0
\]
and 
\begin{equation*}\label{eq_qj}
0 \to A_j\rtimes \Gamma \to  A \rtimes \Gamma
\xrightarrow{\, q_j \,}   C_0(\Xi_j, B) \rtimes \Gamma  \to 0,
\end{equation*}
where we denote by $q_j$ the quotient map  
$A \rtimes \Gamma \to C_0(\Xi_j, B) \rtimes \Gamma$.

So far, the $C^*$-algebra 
$C_0(\Xi_j, B) \rtimes \Gamma$ is not represented in 
$\ell^2(\Gamma,B)$, but a faithful representation on this Hilbert module exists.
Indeed, recall that $\Xi_j=\ol{\{ g^{-1}\cdot \omega_j \mid g \in \Gamma \}}$ for 
some $\omega_j\in \partial \Omega$, where the closure is taken in  $\partial\Omega = \Omega \setminus \Gamma$.
Thus, if $h \in C_0(\Xi_j, B)$, the map 
\[
\Gamma\ni g \mapsto h(g^{-1}\cdot\omega_j ) \in B
\]
defines a function  $\Gamma\to B$.  More generally, for any $f \in A$, the map 
\begin{equation} \label{eq:boundary_j_rep}
\big(\pi_j^0(f)\psi \big)(x) := f( x^{-1}\cdot \omega_j ) \psi(x), \qquad \psi \in \ell^2(\Gamma,B), \,\, x\in \Gamma
\end{equation}
defines a faithful representation of $A|_{\Xi_j}$ into $\End_B\big(\ell^2(\Gamma, B) \big)$
by multiplication operators. 
Thus, by starting with a generating element $\sum_g S_g h_g \in C_0(\Xi_j, B) \rtimes \Gamma$, we can define 
\begin{equation}  \label{eq:quasi-orbit_rep}
   \pi_j: C_0(\Xi_j, B) \rtimes \Gamma \to \End_B( \ell^2(\Gamma, B) ), \qquad 
   \pi_j \Big(\sum_{g\in \Gamma}^{\text{finite}} S_g h_g \Big) = \sum_{g \in \Gamma}^{\text{finite}} S_g \pi_j^0( h_g )
\end{equation}
with $\pi_j^0$ given in Equation \eqref{eq:boundary_j_rep}.  The map $\pi_j$ extends to a faithful representation of $C_0(\Xi_j, B) \rtimes \Gamma$ 
on $\ell^2(\Gamma, B) $ for any $j \in J$. 

\begin{remark}
For any $j \in J$, we think of the representation $\pi_j: C_0(\Xi_j, B) \rtimes \Gamma \to \End_B( \ell^2(\Gamma, B))$ as recovering 
the bulk system at $\Xi_j$, where $\ell^2(\Gamma, B)$ is now considered a bulk $C^*$-module that is considered without reference to 
an interface or other bulk systems. For general  $j\neq j'$, it may not be possible to \emph{simultaneously}  represent  
$C_0(\Xi_j, B) \rtimes \Gamma$ and $C_0(\Xi_{j'}, B) \rtimes \Gamma$ on the {same} $C^*$-module. The reason 
is that the shifts arising from the representations $\pi_j$ and $\pi_{j'}$ may be incompatible with each other.
\end{remark}

With our preparatory work out of the way, the following result can be easily adapted from~\cite{Mantoiu02}.

\begin{prop}[{cf. \cite[Theorem 1.8]{Mantoiu02}}] \label{prop:BEss_spec_decomp}
Let $A$ satisfy Assumption \ref{assump:B-valued_functions}
and consider a covering of $\partial \Omega$ by quasi-orbits satisfying Equation \eqref{eq:covering}.
If $T \in A\rtimes \Gamma \subset \End_B\big(\ell^2(\Gamma, B) \big)$, then 
\[
  \sigma^{B}_{\mathrm{ess}}( T) = \ol{ \bigcup_{j\in J} \sigma\big(q_j(T)\big)}, \qquad \qquad 
  q_j(T) \in C_0(\Xi_j, B) \rtimes \Gamma,
\]
where the overline means the closure of the set in $\C$.
\end{prop}

\begin{remark}
Recalling Remark \ref{rk:Stone-vonNeumann_rk}, Proposition \ref{prop:BEss_spec_decomp} also holds in the  setting of a locally compact and exact group 
$G$ acting trivially on $B$ (which is always true in the case $B=M_N(\C)$, for example).  
In this setting, the exact sequence of interest is 
\[
   0 \to \calK\big( L^2(G) \big) \otimes B \to \big( C_0(\Omega) \rtimes G \big) \otimes B \to \big(C_0(\partial\Omega) \rtimes G \big) \otimes B \to 0
\]
with $\Omega$ a $G$-invariant subset of $\beta G$ and $\partial \Omega = \Omega \setminus G$. We can again take a covering of $\partial \Omega$ 
by quasi-orbits to compute the $B$-essential spectrum, where each $C_0(\Xi_j ) \rtimes G \otimes B$ can be faithfully represented on $L^2(G, B)$.
\end{remark}

\subsection{Non-propagation results}

The algebraic framework we have introduced also can be used to derive so-called  non-propagation results. 
This phenomenon has been studied for Schr\"{o}dinger operators and crossed products by $\R^d$ in~\cite[Section 5]{AMP}. 
See also \cite[Theorem 1.12]{MPR} for the case that a magnetic field is also involved. 
We can adapt the proofs to our studies of discrete interfaces without major obstacles.

Given a vector in our interface system $\psi \in \ell^2(\Gamma, B)$  and 
$H=H^* \in \End_B(\ell^2(\Z, B))$, we can consider the time evolution $e^{itH} \psi$ and spatial propagation 
of $\psi$ as $t \to \infty$. In particular, we may ask if a vector does or does not propagate in space towards a 
given bulk system, represented by a quasi-orbit at infinity. We also consider this question for a discrete time evolution,  
$U^n \psi$ with $U \in \End_B(\ell^2(\Gamma, B))$ unitary, which is of relevance for quantum walk and 
Floquet systems.

For any quasi-orbit $\Xi_j \subset \partial \Omega$, let $\calN_j \subset \beta \Gamma$ be a family of sets of the form $W = \calW\cap \Gamma$, 
where $\calW$ belongs to any  neighbourhood base of $\Xi_j$ in $\Omega \subset \beta\Gamma$. Hence 
$W$ is the restriction to $\Gamma$ of any open set in $\Omega$ containing the subset $\Xi_j$. 
We then consider the indicator function $\chi_{{W}} \in \ell^\infty(\Gamma, B) \subset \End_B(\ell^2(\Gamma, B) )$.

\begin{prop}[{cf. \cite[Section 5, Theorem]{AMP}}] \label{pro_non_pro}
Let $A$ satisfy Assumption \ref{assump:B-valued_functions}
and consider a covering of $\partial \Omega$ by quasi-orbits satisfying Equation \eqref{eq:covering}.
Let $\Xi_j$ be a quasi-orbit and $T \in A \rtimes \Gamma \subset \End_B\big(\ell^2(\Gamma,B) \big)$  a 
normal element.
 If $\eta\in C_0\big(\sigma(T)\setminus \{0\} \big)$ satisfies $\supp(\eta)\cap \sigma\big(q_j(T)\big)=\emptyset$, 
then for any $\varepsilon>0$ there exists $W\in \calN_j$ such that
\[
\big\|\chi_{W} \, \eta(T)\big\|_{\End_B(\ell^2(\Gamma, B))} \leq \varepsilon.
\]
If $T = U$ is unitary, the inequality 
\begin{equation}\label{eq_np}
\big\| \chi_{W} \, U^n\;\!\eta(U)\psi\big\|\leq \varepsilon \|\psi\|
\end{equation}
holds, uniformly in $n\in \Z$ and $\psi \in \ell^2(\Gamma,B)$. If
$T = H$ is self-adjoint, the inequality
\begin{equation}\label{eq_np_2}
\big\| \chi_{W} \, \e^{-itH}\;\!\eta(H)\psi\big\|\leq \varepsilon \|\psi\|
\end{equation}
holds, uniformly in $t\in \R$ and $\psi \in \ell^2(\Gamma,B)$.
\end{prop}

For the first statement, observe that independently the norms of the operators $\chi_{W}$ and $\eta(T)$ have no reason to be small, 
it is only the product of these two operators which turns out to have a small norm. 
The  \emph{non-propagation} comes from Equations \eqref{eq_np} and \eqref{eq_np_2}. Namely,
if the spectral support of $\psi$ with respect to  $U$ (resp.  $H$) does not
meet the spectrum of the bulk operator $q_j(U)$ (resp. $q_j(H)$)  corresponding to the quasi-orbit $\Xi_j$, 
then $\psi$ can not propagate towards this quasi-orbit under the discrete time evolution $U^n \psi$ (resp. $e^{itH} \psi$).
Additional physical explanations are provided in \cite{AMP} in the framework of an evolution defined by a self-adjoint operator
acting on $L^2(\R^d)$.

The proof of Proposition \ref{pro_non_pro} uses the following lemma.
\begin{lemma}[{\cite[Lemma 1]{AMP}}] \label{lemma:function_in_ideal_condition}
Let $J$ be an closed two-sided ideal in a $C^*$-algebra $E$, and let $q: E\to  E/J $ denote the quotient map. 
If $T\in E$ is a normal element and $\eta\in C_0\big(\sigma(T)\setminus\{0\}\big)$ such that $\eta(z)=0$ for 
all $z \in \sigma\big(q(T)\big)$, then $\eta(T)\in J$. 
\end{lemma}

\begin{proof}[Proof sketch of Propostion \ref{pro_non_pro}]
We provide the essential details and  refer to the presentation provided in \cite[Section 5]{AMP} for more details. 
We consider Lemma \ref{lemma:function_in_ideal_condition} in the case 
$E = A \rtimes \Gamma $, $J= A_j\rtimes \Gamma$, where  $A_j$ is defined in Equation \eqref{eq_Aj}, and   $q = q_j$ the quotient map. 
By the assumption on $\eta$ and $T$, it then follows that $\eta(T)\in  A_j\rtimes \Gamma$. 
Considering the algebra $A_j\rtimes \Gamma \subset \End_B\big( \ell^2(\Gamma, B) \big)$, one infers 
that the family $\{\one -\chi_{W}\}_{W\in \calN_j}$ defines
an approximate unit in $\End_B\big( \ell^2(\Gamma, B) \big)$ for $A_j\rtimes \Gamma$, 
which leads directly to the inequality $\big\|\chi_{W} \, \eta(T)\big\|\leq \varepsilon$. 
Equations \eqref{eq_np} and \eqref{eq_np_2}  are then direct consequences as
$U^n$ and $\e^{-itH}$ are unitary operators commuting with $U$ and $H$ respectively. 
\end{proof}

\section{Examples} \label{sec:B-valued_examples}

In all the examples below, unless stated otherwise we will take the internal action $\beta: \Gamma \to \Aut(B)$ to be trivial 
for simplicity. That is,  $\alpha_g( f)(x) = f(g^{-1} x)$ for all $f \in \ell^\infty(\Gamma, B)$ and $x,g \in \Gamma$.

\subsection{Case $B = M_N(\C)$ (Hilbert space setting)}

All of our previous constructions (and those below) can be simplified to the case of a fixed Hilbert space by taking 
$B = M_N(\C)$  with the internal action $\beta: \Gamma \to \Aut\big( M_N(\C)\big)$ trivial. We take   the 
Hilbert space $\calH = \ell^2(\Gamma, \C^N)$, which can be acted on 
by the algebra $\ell^\infty( \Gamma, M_N(\C) ) \rtimes \Gamma$.  As an interface Hilbert space, $\ell^2(\Gamma, \C^N)$ describes 
a mixing of finite-dimensional systems, though we can also consider the spectral analysis 
of band-dominated operators on $ \ell^2(\Gamma, \C^N)$ without reference to discrete interfaces.

We consider a subalgebra  $A \subset \ell^\infty( \Gamma, M_N(\C))$ of multiplication operators  
satisfying Assumption \ref{assump:B-valued_functions}. Then 
$A \cong C_0(\Omega, M_N(\C))$ and any $T \in A \rtimes \Gamma$ is a bounded operator on the Hilbert space $ \ell^2(\Gamma, \C^{N} ) $. 
Hence the $B$-essential spectrum reduces to the usual essential spectrum,
\[
  \sigma^{B}_{\mathrm{ess}}( T) = \sigma_\mathrm{ess}(T) = \ol{ \bigcup_{j\in J} \sigma\big(q_j(T)\big)}, \qquad 
  q_j(T) \in C_0( \Xi_j , M_N(\C) ) \rtimes \Gamma.
\]
Also, for any $j \in J$, the corresponding bulk algebra $C_0( \Xi_j , M_N(\C) ) \rtimes \Gamma$ can be faithfully represented on 
$\ell^2(\Gamma, \C^{N})$.

\subsection{Cartesian anisotropy ($\Gamma = \Z^l$)}  \label{subsec:Cartesian_Aniso}
Let us consider an anisotropy introduced in \cite{Ri05} for  Schr\"{o}dinger operators that extends the domain wall 
example considered in Section \ref{sec:DomainWall}. We let $\ol{\Z} = \Z \cup \{-\infty, \infty\}$ be the closure of 
$\Z$ in $[-\infty, \infty]$, the two-point compactification of $\R$. Then for any $l \in \N$, we let $A_{\ol{\Z}^l} = C\big( \ol{\Z}^l, B \big)$, the 
$B$-valued continuous functions on the discrete hypercube.

To explain  $A_{\ol{\Z}^l} = C\big( \ol{\Z}^l, B \big)$ in more detail,  for each $f \in \ell^\infty( \Z, B)$ and $j\in \{1, \dots,l\}$ we  consider two functions $f_{j\pm}: \Z^{l-1}\to B$ depending 
on all variables except the $j$th variable.
That is, we write $f_{j \pm} (x)$ to denote  a function defined on $\Z^l$ but independent of the component $x_j$ of $x \in \Z^l$.
With this convention in mind, the algebra of Cartesian anisotropy in dimension $l$ is defined by 
\begin{align*}
A_{\ol{\Z}^l}:=\Big\{ f \in \ell^\infty\big(\Z^l,B\big) \, \big| \,  & \text{ for all }  j\in \{1,\dots,d\},\text{ there exists } f_{j\pm}\in 
\ell^\infty \big(\Z^{l-1},B\big) \hbox{ with } \\
&\hspace{3cm} \lim_{r\to \infty} \sup_{x\in \Z^l}\big\|\chi_{[r,\infty)}(\pm x_j)\;\!\big(f(x)- f_{j\pm}(x)\big)\big\|_{B}=0 \Big\}.
\end{align*}
The algebra $A_{\ol{\Z}^l}$ is unital and clearly satisfies Assumption \ref{assump:B-valued_functions}. 
We think of the functions $f_{j \pm}$ as the original function $f$ extended 
to the points $\pm \infty$ in the $j$th coordinate, 
$f_{j \pm}(x) \sim f(x_1,\ldots, x_{j-1}, \pm \infty, x_{j+1}, \ldots, x_l)$.

Geometrically, the boundary at infinity $\ol{\Z}^l \setminus \Z^l$ is the hollow discrete hypercube and each 
face represents a separate bulk system. Looking closer at the algebra, 
there is a $\ast$-homomorphism 
$$
A_{\ol{\Z}^l}\ni f \mapsto \bigoplus_{j\in \{1, \dots,l\}} \big(f_{j-},f_{j+})\in \bigoplus_{j\in \{1, \dots,l\}} 
\Big(A_{\ol{\Z}^{l-1}}^{j-},A_{\ol{\Z}^{l-1}}^{j+} \Big) 
$$ 
 with kernel $C_0(\Z^l, B)$, where 
 $A_{\ol{\Z}^{l-1}}^{j \pm} $ is the algebra $A_{\ol{\Z}^{l-1}}$ for which the $l-1$ variables are 
$x_1, \dots, x_{j-1}, x_{j+1}, \dots,x_l$ (with $x_j$ fixed at $\pm \infty$). This $\ast$-homomorphism is not surjective
since compatibility conditions are not taken into account, meaning  equalities of the type
$$
\lim_{x_k\to +\infty} f_{j+}(x_1,\dots, x_k,\dots, x_{j-1}, x_{j+1}, \dots,x_l) 
= \lim_{x_j\to +\infty} f_{k+}(x_1,\dots, x_{k-1},x_{k+1},\dots, x_{j}, \dots,x_l)
$$
for $j\neq k$. Taking the crossed product, we obtain a $*$-homomorphism
\begin{equation}\label{eq_cart_quotient}
\Big(A_{\ol{\Z}^l} /C_0\big(\Z^l,B\big)\Big)\rtimes \Z^l
\stackrel{q}{\longrightarrow}
\bigoplus_{j\in \{1, \dots, l \}}\Big(C\big(\T_j,A_{\ol{\Z}^{l-1}}^{j-}\rtimes \Z^{l-1}\big),C\big(\T_j,A_{\ol{\Z}^{l-1}}^{j+}\rtimes \Z^{l-1}\big)\Big).
\end{equation}
Indeed, since $A_{\ol{\Z}^{l-1}}^{j \pm}$ depends only on $l-1$ variables, 
 the crossed product by $\Z^l$ can 
be decomposed into a crossed product by $\Z^{l-1}$ and a trivial crossed product by $\Z$ in 
the $j$th direction. 
By the Fourier transform, this trivial crossed product leads to the factor $C(\T_j)$. In order to keep 
track of the special variable in each term of this decomposition, we have indicated it by an index $j$.

The $*$-homomorphism $q$ introduced in Equation \eqref{eq_cart_quotient} leads naturally to an 
explicit formula for the computation of the $B$-essential spectrum and 
non-propagation result. For the $B$-essential spectrum, if $T\in A_{\ol{\Z}^l} \rtimes \Z^l$, let us denote by 
$$
\T_j\ni \theta\mapsto T_{j\pm}(\theta)\in C\big(\T_j, A_{\ol{\Z}^{l-1}}^{j \pm} \rtimes \Z^{l-1}\big)
$$
one component of the image of $q(T)$ in Equation \eqref{eq_cart_quotient}.
Recall that this operator is obtained by firstly considering the image at $\pm \infty$ for the variable $x_j$, 
and then by performing a Fourier transform of the resulting operator, with respect to the same variable. 
Then the general formula for the essential spectrum of $T$ reads
$$
\sigma^{B}_{\mathrm{ess}}(T) = \ol{ \bigcup_{j\in \{1, \dots,l \}}\bigcup_{\theta\in \T_j} 
\sigma\big(T_{j-}(\theta)\big)\cup \sigma\big(T_{j+}(\theta)\big) }.
$$

For the non-propagation properties, let us assume that $T$ is normal and 
$ \ol{\bigcup_{\theta \in \T_j}\sigma\big(T_{j+}(\theta)\big)} \varsubsetneqq \sigma^B_{\mathrm{ess}}(T)$
for some $j \in \{1,\ldots, l\}$.
Choose $\eta\in C(\sigma(T))$ such that $\sigma\big(\eta(T)\big)\cap
\ol{\bigcup_{\theta\in \T_j}\sigma\big(T_{j+}(\theta)\big)} =\emptyset$.
Considering $T=U$ unitary or $T=H$ self-adjoint, then applying 
 Proposition \ref{pro_non_pro},
for any $\varepsilon>0$ 
there exists $r_j\in \Z$ (with $r_j\to \infty$ as $\varepsilon \searrow 0$)  such that
the bounds
\begin{equation*}
\big\|\chi_{[r_j,+\infty)} \;\!U^n\;\!\eta(U)\psi\big\|\leq \varepsilon \|\psi\|, \qquad 
\big\|\chi_{[r_j,+\infty)} \;\!e^{itH}\;\!\eta(H)\psi\big\|\leq \varepsilon \|\psi\|
\end{equation*}
hold, uniformly in $n\in \Z$, $t \in \R$ and $\psi \in \ell^2(\Z^l,B)$. 
An entirely analogous result holds for $ \ol{\bigcup_{\theta \in \T_j}\sigma\big(T_{j-}(\theta)\big)} \varsubsetneqq \sigma^{B}_{\mathrm{ess}}(T)$.

When $\beta: \Z^l \to \Aut(B)$ is a non-trivial action, there is an injective $\ast$-homomorphism 
\[
\Big(A_{\ol{\Z}^l} /C_0\big(\Z^l,B\big)\Big)\rtimes \Z^l
\stackrel{q}{\longrightarrow}
\bigoplus_{j\in \{1, \dots, l \}} \big( A_{\ol{\Z}^{l-1}}^{j-} \rtimes \Z^{l}, \, A_{\ol{\Z}^{l-1}}^{j+} \rtimes \Z^{l-1}\big).
\]
Because $\Z_j$ can act non-trivially on $B$, we can not always further decompose via the Fourier transform.

\subsection{Vanishing oscillation}
The algebra of functions with vanishing oscillations\ was considered in~\cite{DS85, Mantoiu02} and 
is related to  the Higson compactification of coarse 
metric spaces, see~\cite[Section 2.3]{Roe03}.
We define the algebra 
$$
A_v:=\big\{ f \in  \ell^\infty(\Gamma, B) 
\mid  \alpha_g(f) - f \in C_0(\Gamma, B) \hbox{ for any }g \in\Gamma \big\},
$$
which is is unital and easily satisfies Assumption \ref{assump:B-valued_functions}. 
The  compactification $\Omega$ of $\Gamma$ such that $A_v = C(\Omega, B)$ can not be easily 
described.
However, we have that each point $\omega \in \partial \Omega = \Omega\setminus \Gamma$ is invariant
under the action of $\Gamma$, precisely because $\alpha_g(f) - f \in C_0(\Gamma, B)$ 
for any $f\in C(\Omega, B)$. As a consequence, the action of $\Gamma$ on $\partial \Omega$ is trivial and 
each quasi-orbit consists of a singleton $\{\omega\} \subset \partial \Omega$. Therefore 
\[
   C( \Xi_\omega, B) \cong  C( \{\omega\}, B) \cong B , \qquad \qquad 
   C( \Xi_\omega, B) \rtimes \Gamma \cong B \rtimes \Gamma \cong C(\wh{\Gamma}, B ),
\]
where we have used that the action of $\Gamma$ on $B$ is trivial and  $\wh{\Gamma}$ is the 
Pontryagin dual, $C^*(\Gamma) \cong C\big(\wh{\Gamma})$. Similarly, 
\[
A_v / C_0(\Gamma, B) = C( \partial \Omega, B)  
\subset \prod_{\omega\in \partial\Omega} B, \qquad \qquad 
\Big( A_v / C_0(\Gamma, B) \Big) \rtimes \Gamma \cong C\big( \partial\Omega\times \wh{\Gamma}, B\big).
\]

If $f \in C(\Omega, B)$, the values taken by $f$ on $\partial\Omega$ correspond to its asymptotic range,
\[
f(\Gamma, B)_{\rm asy} :=\overline{\bigcap_{K\in \calK(\Gamma)} f(\Gamma\setminus K)},
\] 
where $\calK(\Gamma)$ denotes the set of all compact (finite) subsets of $\Gamma$ 
and the closure is in $B$.

Let us consider the  $B$-essential spectra in a simple setting.
If $T = S_g f \in A_v \rtimes \Gamma$ with $f \in A_v$, then 
\[
\sigma^B_{\rm ess}(S_g f) = \ol{ \bigcup_{h \in  f(\Gamma, B)_{\rm asy}} \sigma(S_g h) }.
\]

For non-propoagation, we consider $f \in A_v$ such that 
$U = S_g f \in A_v \rtimes \Gamma$ is unitary. Then for 
$h \in  f(\Z^l, B)_{\rm asy}$, we consider the case that 
 $\sigma(S_g h)\varsubsetneqq \sigma_{\rm ess}^B(U)$. 
Taking a function  $\eta\in C(\T)$ such that $\sigma (\eta(U) )\cap \sigma(S_h g) =\emptyset$, 
 Proposition \ref{pro_non_pro} implies that 
for any $\varepsilon>0$, 
there exists  a sequence of discrete subsets $(W_j)_{j\in \N}\subset \Gamma$ with $|W_j|\to \infty$  such that
\begin{equation*}
\big\|\chi_{W_j} \;\!U^n\;\!\eta(U)\psi\big\|\leq \varepsilon \|\psi\|
\end{equation*}
holds for all $j$, uniformly in $n\in \Z$ and $\psi \in \ell^2(\Gamma, B)$. Observe that the sequence of subsets $W_j$ 
is going to infinity and become more sparse as $j$ increases.

If $\beta: \Gamma \to \Aut(B)$ is a non-trivial action, we have that 
\[
   \big( A_v / C_0(\Gamma, B) \big) \rtimes \Gamma \cong C\big( \partial\Omega, B \rtimes \Gamma \big)
\]
and, like the case of the trivial action,
\[
\sigma^B_{\rm ess}(S_g f) = \ol{ \bigcup_{h \in  f(\Gamma, B)_{\rm asy}} \sigma(S_g h) }.
\]

\subsection{Radial symmetry ($\Gamma = \Z^l$)}  \label{subsec:radial_symm}
 
We start  by describing a compactification of $\Z^l$ and then consider the 
 algebra of $B$-valued functions. 
 We let $\S^{l-1} \subset \R^l$ denote the unit sphere. 
For any open set $\calO\subset \S^{l-1}$, we introduce the truncated cone
$\frakC(\calO,r) \subset \Z^l$,
\begin{equation}\label{eq_tr_cone}
\frakC(\calO,r):=\big(\frakC(\calO)\setminus  \ol{B(\mathbf{0}, r )} \big)\cap \Z^l,
\end{equation}
where $\ol{B(\mathbf{0}, r )} \subset \R^l$ denotes the closed ball with center $0$ and 
radius $r$ and  $\frakC(\calO)$ denotes the infinite open cone in $\R^l$ of apex $0$ that intersects $\S^{l-1}$ on $\calO$.
With these notations, we introduce the compact space $\Omega:=\Z^l\sqcup \S^{l-1}$ made of $\Z^l$ 
together with a sphere $\S^{l-1}$ added \emph{at infinity}.
For its topology, a neighbourhood basis for any open set $\calO\subset \S^{l-1}$ 
is made of elements of the form
$\big(\frakC(\calO,r),\calO\big)\subset \Z^l \sqcup \S^{l-1}$. 
In particular,  if $(x_n)_{n\in \N}$ is a sequence in $\Z^l$ converging to the point 
$\omega\in \S^{l-1}$  at infinity, then for any $y \in \Z^l$ the shifted sequence $(x_n+y)_{n\in \N}$ converges 
to the same element $\omega$. As a consequence of this observation, 
any point on the sphere at infinity is left invariant by the translation action on $\Omega$.

Having described the compact space $\Omega$, we now consider 
$A_r:=C(\Omega,  B)$,  the algebra of functions $f :\Z^l \to B$ 
such that $f$ admits a continuous extension on $\Omega$. 
As one can easily check, this algebra satisfies Assumption \ref{assump:B-valued_functions}.
We also remark that the radially symmetric functions are a special case of functions with 
vanishing oscillation for $\Gamma = \Z^l$, $A_r \subset A_v$.

By construction, the boundary at infinity $\partial \Omega = \S^{l-1}$ and has a trivial $\Z^l$-action. Therefore 
\[
A_r/C_0(\Z^l, B)= C(\S^{l-1}, B), \qquad \quad 
\big(A_r/C_0(\Z^l, B) \big)\rtimes \Z^l \cong C(\S^{l-1}\times \T^l,B ),
\]
where $\S^{l-1}\times \T^l$ is endowed with the product topology.
We can therefore easily compute the $B$-essential spectrum 
\[
  \sigma_{B, \mathrm{ess}}(T) = \ol{\bigcup_{\omega \in \S^{l-1}} \bigcup_{\xi\in \T^l} 
\sigma\big(q(T)(\omega,\xi)\big)}, \qquad 
  T \in A_r \rtimes \Z^l, \qquad q(T)(\omega,\xi) \in B.
\]

Finally, if $T$ is normal and there exists $\omega\in \S^{l-1}$ such that
$\ol{\bigcup_{\xi\in \T^l} \sigma\big(q(T)(\omega,\xi)\big)}  \varsubsetneqq \sigma^{B}_{\mathrm{ess}}(T)$, 
then we can fix $\eta\in C(\sigma(T))$ such that $\sigma\big(\eta(T)\big)\cap
\ol{\bigcup_{\xi\in \T^l} \sigma\big(q(T)(\omega,\xi)\big)} =\emptyset$.
Applying Proposition \ref{pro_non_pro} for $T=U$ unitary or $T=H$ self-adjoint, 
for any $\varepsilon>0$ 
there exists a truncated cone $\frakC(\calO,r)$ with $\calO$ an neighbourhood of $\omega$ in $\S^{l-1}$ such that
the inequalities
\begin{equation*}
\big\|\chi_{ \frakC(\calO,r)} \;\!U^n\;\!\eta(U)\psi\big\|\leq \varepsilon \|\psi\|, \qquad 
\big\|\chi_{\frakC(\calO,r)} \;\!e^{itH}\;\!\eta(H)\psi\big\|\leq \varepsilon \|\psi\|
\end{equation*}
hold uniformly for $n\in \Z$, $t\in \R$ and $\psi \in \ell^2(\Z^l,B)$.

Because $\Z^l$ acts trivially on the boundary, if $\beta: \Z^l \to \Aut(B)$ is non-trivial, then
\[
  \big(A_r/C_0(\Z^l, B) \big)\rtimes \Z^l \cong C( \S^{l-1}, B \rtimes \Z^l), \qquad 
  \sigma_{B, \mathrm{ess}}(T) = \ol{\bigcup_{\omega \in \S^{l-1}} \sigma\big( q(T)(\omega) \big)}, 
\]
where $q(T)(\omega) \in  B \rtimes \Z^l $ for all $\omega \in \S^{l-1}$.

\subsection{Functions asymptotically supported on disjoint cones} \label{subsec:asymp_cone}

Let us consider a non-unital example via 
a subalgebra of the algebra $A_r$ introduced in the previous subsection. 
Let $\{\ol{\calO_1}, \ldots, \ol{\calO_M}\} \subset \S^{l-1}$ be a finite family of closed 
subsets such that $\ol{\calO_j} \cap \ol{\calO_k} = \emptyset$ for all $j\neq k$. 
The subalgebra $A_c$ consists of $f \in A_r$ vanishing on $\Omega\setminus \bigsqcup_{j=1}^M \ol{\calO_j}$.
Equivalently, $A_c$ contains bounded  functions $f: \Z^l \to B$ admitting
a continuous extension to $\Z^l \sqcup \S^{l-1}$ and such that for any $\epsilon>0$ there exists $r>0$ with 
\[
\| f(x) \|_{B} \leq \epsilon \quad 
 \text{for all } x \in \Z^l \setminus \bigsqcup_{j=1}^M \frakC(\ol{\calO_j},r) 
\]
and $\frakC(\ol{\calO_j},r)$ defined as in Equation \eqref{eq_tr_cone}. 
Fixing a  disjoint closed family, $\{\ol{\calO_1}, \ldots, \ol{\calO_M}\} \subset \S^{l-1}$, 
the algebra $A_c$ of functions asymptotically localised on these cones is a translation invariant subalgebra 
of $A_r$ satisfying Assumption \ref{assump:B-valued_functions}.
The space $\Omega \subset \beta\Z^l$ such that $A_c \cong C_0(\Omega, B)$ is locally compact and 
with $\Omega \setminus \Z^l \cong \bigsqcup_{j=1}^M \calO_j$. Therefore we also have that 
\[
A_c / C_0 (\Z^l, B ) =    \bigoplus_{j=1}^M C_0(\calO_j, B).
\]
Considering the crossed product by $\Z^l$, the  points $\omega \in \calO_j \subset \S^{l-1}$ are $\Z^l$-invariant and we can further simplify 
\[
\big( A_c / C_0 (\Z^l, B ) \big) \rtimes \Z^l \cong \bigoplus_{j=1}^M C_0(\calO_j, B) \rtimes \Z^l 
  \cong \bigoplus_{j=1}^M C_0(\calO_j \times \T^l, B) 
\]
We can therefore compute the $B$-essential spectrum 
\[
  \sigma_{B, \mathrm{ess}}(T) = \ol{ \bigcup_{j=1}^M \bigcup_{\omega_j \in \calO_j} \bigcup_{\xi\in \T^l} 
\sigma\big(q(T)(\omega,\xi)\big)}, \qquad 
  T \in A_c \rtimes \Z^l.
\]
The non-propagation properties of unitary or self-adjoint elements in $A_c \rtimes \Z^l$ are entirely 
analogous to the case of the radially symmetric functions from Section \ref{subsec:radial_symm}, so we omit the details.

\subsection{Examples from discrete spaces with a cocompact  $\Gamma$-action}

We briefly outline how our framework can also be applied to interfaces that are not a group, but come with cocompact (cofinite) 
group action. Namely, we consider $\ell^2(X, B)$, where $X$ is discrete and there is an action 
\[
   \Gamma \times X \to X, \qquad \Gamma \backslash X \cong X_0, \qquad |X_0| < \infty.
\]
An important example comes from topological crystals, where $X_0$ describes the points in the unit cell and $\Gamma \cong \Z^l$ for 
some $l\geq 1$. See~\cite{Sunada} for further details and many examples. 

Because the action of $\Gamma$ on $X$ is cocompact, there is a natural identification 
$\ell^2(X, \C) \cong \ell^2 \big(\Gamma, \C^{|X_0|} \big)$, which also extends to the Hilbert $C^*$-module 
\[
   \ell^2(X, B) \cong \ell^2\big( \Gamma, B^{|X_0|} \big).
\]
If $B$ has an internal $\Gamma$-action $\beta: \Gamma \to \Aut(B)$, then this extends diagonally to an action 
of $\Gamma$ on $B^{|X_0|}$. Hence, taking $\tilde{B} = B^{|X_0|}$,  our interface $C^*$-module is 
$\ell^2(\Gamma, \tilde{B})$ and we are back in the setting of 
Section \ref{sec:CStar_mod_interface}. Therefore our framework and results can also be applied to topological crystals and other 
discrete spaces with a cocompact $\Gamma$-action.

\subsection{Examples from open or closed $\Gamma$-invariant subsets at infinity} \label{subsec:Pullback_example}

Given a locally compact subset $\Omega \subset \beta \Gamma$ such that $A = C_0(\Omega, B)$ satisfies 
Assumption \ref{assump:B-valued_functions}, the various bulk systems and the $B$-essential spectrum of elements in $C_0(\Omega, B)\rtimes \Gamma$ 
are determined by the subset $ \partial \Omega \subset \beta\Gamma \setminus \Gamma$. Here we go the other 
direction. Suppose that $Z$ is a $\Gamma$-invariant and locally compact subset of $\beta\Gamma \setminus \Gamma$. 
We also assume that $Z$ is an open or closed subset of 
$\beta \Gamma \setminus \Gamma$ and treat these cases separately.

If $Z \subset \beta\Gamma \setminus \Gamma$ is open, then we can extend any function $h \in C_0( Z, B)$ by $0$ to obtain 
a function $h \in C(\beta\Gamma \setminus \Gamma, B)$ and injective $\ast$-homomorphism 
$\tau: C_0(Z, B)\to C(\beta\Gamma \setminus \Gamma,  B)$. We consider the algebra 
\[
  A_Z = \big\{ f \in C(\beta \Gamma, B) \mid q_\beta(f) \in C_0(Z) \big\}, \qquad q_\beta: C(\beta \Gamma, B) \to C(\beta\Gamma \setminus \Gamma, B),
\]
which naturally fits in the short exact sequence 
\[
  0 \to C_0(\Gamma, B) \to A_Z \to C_0(Z, B) \to 0.
\]
Furthermore, $A_Z \cong C_0( \wt{Z}, B)$ with $\wt{Z} \cong q_\beta^{-1}(Z)$ an open subset of $\beta \Gamma$.

If $Z$ is a closed subset of $\beta\Gamma \setminus \Gamma$, then any function $h\in C(Z, B)$ can be 
extended to a function in $C(\beta\Gamma \setminus \Gamma, B)$ 
by the Tietze Extension Theorem. 
Fixing such an extension, we obtain an injective $\ast$-homorphism $\tau: C(Z, B) \to C(\beta\Gamma \setminus \Gamma, B)$. 
We can again build a short exact sequence via the pullback $C^*$-algebra, 
\[
  0 \to C_0(\Gamma, B) \to A_Z \to C(Z, B) \to 0,
\]
where 
\[
   A_Z = \big\{ (f, h) \in C(\beta\Gamma, B) \oplus C(Z, B) \mid q_\beta(f) = \tau(h) \big\}
\]
and the map $C_0(\Gamma, B) \to A_Z$ is given by $f\mapsto (f, 0)$.

For $Z$ open or closed, $A_Z$ has a $\Gamma$-action and fits into 
our general framework. In particular,  we have a short exact sequence of crossed products,
\[
   0 \to \calK\big(\ell^2(\Gamma)\big) \otimes B \to A_Z \rtimes \Gamma \to C_0(Z, B)\rtimes \Gamma \to 0.
\]

While this pullback construction is rather ad-hoc,   given a single or collection of bulk systems 
spatially  localised in the 
set $Z \subset \beta\Gamma \setminus \Gamma$ at infinity, we can find an algebra $A_{Z}$ of multiplication 
operators on $\End_B(\ell^2(\Gamma, B))$ whose crossed product by $\Gamma$ 
 will give a model of the dynamics on the interface $\ell^2(\Gamma, B)$. The advantage of the algebra $A_{Z}$ over the 
universal algebra $C(\beta \Gamma, B) \rtimes \Gamma$ is that the quotient of $A_{Z}\rtimes \Gamma$ by the compact operators 
$\calK( \ell^2(\Gamma) )\otimes B$ will precisely take us to $C_0(Z, B) \rtimes \Gamma$, the bulk system  of interest to us.

The $B$-essential spectrum of any element $T \in A_Z \rtimes \Gamma$ is given by the spectrum of 
$q(T) \in C_0(Z, B) \rtimes \Gamma$. Though in general  $C_0(Z, B) \rtimes \Gamma$ might not have a 
faithful representation on $\ell^2(\Gamma, B)$. This can be rectified by taking a covering of $Z$ by quasi-orbits 
$\{\Xi_j\}_{j\in J}$ where each $C_0(\Xi_j, B) \rtimes \Gamma$ can be faithfully represented on the interface $C^*$-module 
$\ell^2(\Gamma, B)$. Indeed, the existence of the canonical representation from Equation \eqref{eq:quasi-orbit_rep} is one 
of the key reasons we consider a covering of the boundary at infinity by quasi-orbits.

\section{Interfaces of gapped systems and index theory} \label{sec:Interface_index}

So far we have established a rather general framework to consider discrete interfaces described by a discrete abelian group $\Gamma$ 
and shown some spectral consequences of this algebraic approach. Our general model uses  an algebra $A$ satisfying Assumption \ref{assump:B-valued_functions}
 and whose  spatial asymptotics  is  used to access the various bulk systems at infinity. 

Building from our knowledge of topological phases of matter, we are often interested in bulk systems possessing a spectral gap. 
The physical principle of the bulk-boundary/defect/interface correspondence then says that topological properties of the interface 
should be closely linked to a bulk topological invariant at infinity. 

In this section, our aim is to define an interface index taking values in the $K$-theory group of the ambient ($\Gamma$-defect) $C^*$-algebra $B$. 
Using the interface extension 
\[
   0 \to \calK( \ell^2(\Gamma)) \otimes B \to C_0(\Omega, B)\rtimes  \Gamma \to C_0( \partial \Omega, B )\rtimes \Gamma \to 0,
\]
we then relate this interface index to quantities determined purely by the bulk system at infinity. In general there no guarantee of a 
one-to-one equivalence between bulk and interface indices, but certainly a non-trivial interface index must imply non-trivial bulk indices.
We also show that a $\Gamma$-equivariant decomposition of the boundary at infinity $\partial \Omega$ leads to a decomposition of the interface index.

This section will assume some  knowledge of Fredholm operators on $C^*$-modules and the theory of 
 extensions of $C^*$-algebras. A brief overview is given in Appendices \ref{sec:KK} and \ref{sec:extensions}.

\subsection{Fredholm operators in interface systems}

Recall that our interface system is described by the $C^*$-module $\ell^2(\Gamma, B)$ with dynamics modelled by 
an algebra $C_0( \Omega, B) \rtimes \Gamma \subset \End_B( \ell^2(\Gamma, B) )$.
An operator $F \in \End_B( \ell^2(\Gamma, B) )$ is Fredholm if $q(F) \in \calQ_B( \ell^2(\Gamma, B) )$ is invertible. That is, 
\[
   F \ \  \text{Fredholm} \quad \Longleftrightarrow \quad 0 \notin \sigma^{B}_{\mathrm{ess}}( F).
\]

We can therefore use our results on the $B$-essential spectrum to infer Fredholm properties of interface operators. 

\begin{prop} \label{prop:uniformly_gapped_bulk_gives_Fredholm}
Let $A \cong C_0(\Omega, B)$ satisfy Assumption \ref{assump:B-valued_functions}
and consider a covering of $\partial \Omega$ by quasi-orbits satisfying Equation \eqref{eq:covering}.
If there is some $F \in \big(C_0(\Omega, B) \rtimes \Gamma \big)^{\sim}$ and $\epsilon > 0$ such that 
$\sigma\big( q_j(F) \big) \cap (-\epsilon, \epsilon)=\emptyset$ for all $j \in J$, 
 then $F$ is Fredholm on $\ell^2(\Gamma,B)$.
\end{prop}
\begin{proof}
Because  there is a uniform bound on the spectrum of $q_j(F)$ away from $0$, Proposition \ref{prop:BEss_spec_decomp} 
implies that 
$0 \notin \ol{\bigcup_{j \in J} \sigma\big(q_j(F)\big)} = \sigma^{B}_{\mathrm{ess}}( F)$.
\end{proof}

Recall that we can think of the quasi-orbits $C_0(\Xi_j, B) \rtimes \Gamma$ as representing different bulk 
systems at infinity. Proposition \ref{prop:uniformly_gapped_bulk_gives_Fredholm} says that if all our bulk systems are 
\emph{uniformly} gapped, then they can be pulled back to a single Fredholm operator on the interface system $\ell^2(\Gamma,B)$.
Clearly if there are only finitely many quasi-orbits, then the uniformity of the spectral gap is  satisfied.

We can also go in the other direction. Namely, a Fredholm operator on $\ell^2(\Gamma, B)$ will give invertible (gapped) operators 
in each bulk system.

\begin{lemma} \label{lemma:Fredholm_is_invertible_on_quasi_orbit}
Let $A \cong C_0(\Omega, B)$ satisfy Assumption \ref{assump:B-valued_functions}
and consider a covering of $\partial \Omega$ by quasi-orbits satisfying Equation \eqref{eq:covering}.
If $F \in \big(C_0(\Omega, B) \rtimes \Gamma \big)^{\sim}$ is
Fredholm, then for all $j \in J$, $q_j(F) \in \big( C_0(\Xi_j)\rtimes \Gamma \big)^{\sim}$ 
is invertible.
\end{lemma}
\begin{proof}
Because $F$ is Fredholm, $0 \notin \sigma^{B}_{\mathrm{ess}}(F)$. The result then 
follows from Proposition \ref{prop:BEss_spec_decomp}.
\end{proof}

\subsection{The Busby invariant for the interface short exact sequence}

Recalling Assumption \ref{assump:B-valued_functions}, we have an algebra $A = C_0(\Omega, B)$, which fits in to the 
short exact sequence 
\begin{equation} \label{eq:Interface_extension_later}
   0 \to C_0( \Gamma, B) \rtimes \Gamma  \to C_0( \Omega, B) \rtimes \Gamma \xrightarrow{ \, q \, } C_0( \partial \Omega, B) \rtimes \Gamma \to 0,
\end{equation}
where $\partial\Omega = \Omega \setminus \Gamma$. Our aim is to study $K$-theoretic quantities from this short exact sequence. We can 
connect a short exact sequence to $K$-theory provided there exists a positive semi-splitting, a completely positive and contractive map 
$\rho: C_0( \partial \Omega, B) \rtimes \Gamma \to C_0( \Omega, B) \rtimes \Gamma$ such that $q \circ \rho = \mathrm{Id}$. To ensure 
the existence of this semi-splitting,  we 
require further assumptions on our algebra $C_0( \Omega, B)$ and space $\Omega$.

\begin{assumption}  \label{assump:Extension_works}
Let $A \cong C_0(\Omega, B) \subset \ell^\infty(\Gamma, B)$ be a $C^*$-algebra satisfying Assumption \ref{assump:B-valued_functions}. 
We further assume that  the quotient $\partial \Omega = \Omega \setminus \Gamma$ is second countable.
\end{assumption}

The assumption that $\partial \Omega$ is second countable is essential. For example, the short exact sequence 
\[
   0 \to C_0(\Z) \to \ell^\infty( \Z) \to \ell^\infty(\Z)/ C_0(\Z) \to 0
\]
is not semi-split. Hence we restrict to algebras $A \cong C_0(\Omega, B)$ such that 
$C_0(\partial \Omega, B)$ is separable. 

Assumption  \ref{assump:Extension_works} is sufficient to obtain a completely 
positive and contractive semi-splitting 
$\rho: C_0( \partial \Omega, B) \rtimes \Gamma \to C_0( \Omega, B) \rtimes \Gamma \subset \End_B( \ell^2(\Gamma, B) )$ 
such that $q \circ \rho = \mathrm{Id}$. The existence of this semi-splitting is well-known as a special case of the 
Choi--Effros Lifting Theorem~\cite{ChoiEffros}, though we review the construction as it will be useful to us later.

We construct a semi-splitting in steps. First, we look at the short exact sequence 
\[
   0 \to C_0(\Gamma) \to C_0(\Omega) \to C_0(\partial \Omega) \to 0
\]
with $C_0(\Gamma)$ an essential ideal in $C_0(\Omega)$. Hence we can consider 
$C_0(\Omega) \subset \Mult( C_0(\Gamma) ) \cong C(\beta \Gamma)$ and $C_0(\partial \Omega) \subset C( \beta \Gamma \setminus \Gamma)$.

\begin{lemma} [{cf. \cite[Proposition 2.4.2]{BrownOzawa}}] \label{lemma:scalar_cp_splitting}
There is a completely positive and contractive map $\rho_{0}: C_0(\partial \Omega) \to C(\beta \Gamma)$ such that $q \circ \rho_{0} = \mathrm{Id}$. 
Furthermore, we can take $\rho_0$ to be a $\Gamma$-equivariant map.
\end{lemma}
\begin{proof}[Proof sketch]
We will consider the case that $\partial \Omega$ is closed and therefore compact. The case of $\partial\Omega \subset \beta\Gamma \setminus \Gamma$ 
open can be dealt with via the considerations in~\cite[Section 2.2]{BrownOzawa}. 
The first statement  will follow if we can define a completely positive contraction  $\tilde{\rho}_{0}: C(\partial \Omega) \to C(\beta \Gamma)$ such that for any 
finite set $\calF \subset C(\partial \Omega)$ and $\epsilon >0$, $\| f - q \circ \tilde{\rho}_{0}(f) \| < \epsilon$ for all $f \in \calF$. 
Fixing $\calF$ and $\epsilon$, we can take a finite cover $\{U_1,\ldots, U_n\}$ of $\partial \Omega$  such that for each $f \in \calF$ and $j=1,\ldots , n$, 
$|f(x) - f(y) | < \epsilon$ for any $x,y \in U_j$. We take a partition of unity $\{\chi_1, \ldots, \chi_n\}$ subordinate to this cover and fix the points 
$y_j \in U_j$ for all $j=1,\ldots, n$. Because $\chi_j \in C(\partial \Omega)$ is positive, it has a positive lift 
$\tilde{\chi}_j \in C(\Omega) \subset C(\beta \Gamma)$ with $\| \tilde{\chi}_j \| \leq \| \chi_j\|$. We then define 
\[
   \tilde{\rho}_{0}(f) = \sum_{j=1}^n f(y_j) \tilde{\chi}_j,
\]
which is a completely positive contraction such that 
\begin{align*}
  \big\| f - q\circ \tilde{\rho}_{0}(f) \big\|  &= \Big\| \sum_j  \chi_j ( f - f(y_j) )  \Big\| < \epsilon \, \Big\| \sum_j \chi_j \Big\| = \epsilon
\end{align*}
as required.

Because $q: C(\Omega) \to C(\partial \Omega)$ is a $\Gamma$-equivariant map, we have that 
$q\big(\alpha_g^\Omega ( \tilde{\chi}_j ) \big) = \alpha_g^{\partial \Omega} (\chi_j)$ for all $g \in \Gamma$ and 
$j =1,\ldots, n$. Hence 
 $\alpha_g^\Omega \circ \tilde{\rho}_0 - \tilde{\rho}_0 \circ \alpha_g^{\partial \Omega} \in C_0(\Gamma)$ for all $g \in \Gamma$ 
but does not imply that $\tilde{\rho}_0$ is $\Gamma$-equivariant. 
Because $\Gamma$ is amenable, we can take an average over the semi-splitting 
$\tilde{\rho}_{0}$ to obtain a $\Gamma$-equivariant semi-splitting. Namely, we define
\[
  \rho_0: C_0( \partial \Omega) \to C(\beta \Gamma), \qquad  \quad
  \rho_0 = \lim_{\Lambda \nearrow \Gamma} \frac{1}{|\Lambda|} \sum_{g\in \Lambda} \alpha_g \circ \tilde{\rho}_{0}, 
\]
which is a completely positive contraction such that $q \circ \rho_0 = \mathrm{Id}$ and 
\[
  \alpha_g^\Omega \big( \rho_0(f) \big) = \rho_0\big( \alpha_g^{\partial \Omega} (f) \big)
\]
for any $f \in C_0(\partial \Omega)$. 
\end{proof}

Because $q \circ \rho_0 = \mathrm{Id}$ with $q$ a $\ast$-homomorphism, the map 
$\rho_0$ is a $\ast$-homomorphism modulo the ideal $C_0(\Gamma)$, 
\[
   \rho_0 (f_1) \rho_0(f_2) = \rho_0(f_1 f_2) + C_0(\Gamma).
\]

Having constructed a completely positive semi-splitting $\rho_0 : C_0(\partial \Omega) \to C( \beta \Gamma)$, there is a canonical extension 
\[
   \rho: C_0( \partial \Omega, B) \rtimes \Gamma \to C_0(\Omega, B) \rtimes \Gamma \subset \End_B(\ell^2(\Gamma, B)) ,
\]
such that for generating elements 
\[
    \rho\Big( \sum_g S_g (f_g \otimes b) \Big) = \sum_g S_g ( \rho_0 (f_g) \otimes b ).
\]
Indeed for $f \in C(\partial \Omega, B)$, 
\[
   q \circ \rho ( S_g f) = S_g q( \rho_0( f) ) = S_g f
\]
and 
\begin{align*}
    \rho( S_{g_1} f_1 ) \rho( S_{g_2} f_2) &= S_{g_1} \rho_0( f_1) S_{g_2} \rho_0(f_2) = S_{g_1 g_2} \alpha_{g_2} ( \rho_0( f_1) ) \rho_0(f_2) \\
    &= S_{g_1 g_2} \rho_0( \alpha_{g_2} (f_1) ) \rho_0(f_2) = S_{g_1 g_2} \rho_0( \alpha_{g_2}(f_1) f_2 ) + S_{g_1 g_2} f_0 \\
    &=  \rho( S_{g_1 g_2} \alpha_{g_2} (f_1) f_2 ) +S_{g_1g_2} f_0  = \rho( S_{g_1 } f_1 S_{g_2} f_2 ) + S_{g_1g_2} f_0
\end{align*}
for any $g_1,g_2 \in \Gamma$, $f_1, f_2 \in C_0(\partial \Omega, B)$ and 
where $S_{g_1g_2} f_0 \in C_0(\Gamma, B)\rtimes \Gamma$.

We have shown the following. 
\begin{lemma}
Let $A \cong C_0(\Omega, B)$ be a $C^*$-algebra satisfying Assumption \ref{assump:Extension_works}. Then the short exact sequence 
\[
   0 \to \calK \otimes B \to C_0(\Omega, B) \rtimes \Gamma \to C_0( \partial \Omega, B) \rtimes \Gamma \to 0
\]
is semi-split.
\end{lemma}

We can now invoke the key results of Sections \ref{subsec:semi-split_inv_ext} and \ref{subsec:ext_boundary_map}, which we summarise.

\begin{prop} \label{prop:Busby_of_interface_extension}
Let $A \cong C_0(\Omega, B)$ be a $C^*$-algebra satisfying Assumption \ref{assump:Extension_works}. Then there is a 
$\ast$-homomorphism $\pi: C_0( \partial \Omega, B) \rtimes \Gamma \to \End_B(\ell^2(\Gamma, B)^{\oplus 2} )$ and 
$P \in \End_B( \ell^2(\Gamma, B)^{\oplus 2} )$ such that 
\[
  \tau_{\partial \Omega} : C_0( \partial \Omega, B) \rtimes \Gamma \to \calQ_B( \ell^2(\Gamma, B) ), \qquad 
  \tau_{\partial \Omega}(a) = q\big( P \pi(a) P \big)
\]
defines an invertible extension $[\tau_{\partial \Omega}] \in \Ext^{-1}( C_0( \partial \Omega, B) \rtimes \Gamma, B \otimes \calK )$ 
that represents the short exact sequence from Equation \eqref{eq:Interface_extension_later}.
Furthermore, for any invertible $w \in (C_0( \partial \Omega, B) \rtimes \Gamma)^\sim$, the element 
$P \pi(w) P \in \End_B( \ell^2(\Gamma, B)^{\oplus 2} )$ is Fredholm on the submodule $P\cdot  \ell^2(\Gamma, B)^{\oplus 2} \cong \ell^2(\Gamma, B)$.
\end{prop}

\subsection{The interface index}

Given a Fredholm operator $F \in \End_B( \ell^2(\Gamma, B))$, the triple 
\begin{equation} \label{eq:Basic_interface_Kasmod}
    \Big( \C, \, \ell^2(\Gamma, B) \otimes \C^2, \, \tilde{F}  \Big), \qquad \tilde{F} = \chi \begin{pmatrix} 0 & F^* \\ F & 0 \end{pmatrix}
\end{equation}
is a Kasparov module (see Example  \ref{ex:BasicKasMod}), where $\ell^2(\Gamma, B) \otimes \C^2$ is graded by the self-adjoint unitary 
$\left( \begin{smallmatrix} \one & 0 \\ 0 & -\one \end{smallmatrix}\right)$ and 
$\tilde{F} =\chi \left( \begin{smallmatrix} 0 & F^* \\ F & 0 \end{smallmatrix}\right)$ is a self-adjoint, odd and Fredholm operator on 
$\ell^2(\Gamma, B) \otimes \C^2$ such that $\one -  \tilde{F}^2 \in \mathbb{K}_B( \ell^2(\Gamma, B) \otimes \C^2)$. 
Hence the triple from Equation \eqref{eq:Basic_interface_Kasmod} gives 
an equivalence class $[F] \in KK(\C, B)$.

\begin{defn}
Let $\ell^2(\Gamma, B)$ be a $C^*$-module describing an discrete interface. If the adjointable operator $F \in \End_B( \ell^2(\Gamma, B))$ is Fredholm, 
then the interface index is given by $[F] \in KK(\C, B) \cong K_0(B)$, the equivalence class  of the 
Kasparov module from Equation \eqref{eq:Basic_interface_Kasmod}.
\end{defn}

Proposition \ref{prop:uniformly_gapped_bulk_gives_Fredholm} and  \ref{prop:Busby_of_interface_extension} show that the 
interface index is well-defined for a large class of examples with gapped bulk systems.
Indeed, let us  consider dynamics on the interface $\ell^2(\Gamma, B)$ described by $F \in \big(C_0(\Omega, B) \rtimes \Gamma \big)^{\sim}$. 
Given  a covering 
of $\partial \Omega$ by quasi-orbits $\{\Xi_j\}_{j\in J}$ such that  $q_j(F) \in \big(C_0(\Omega, B) \rtimes \Gamma \big)^{\sim}$ is uniformly gapped 
for all $j \in J$, then the interface index $[F] \in KK(\C, B)$ is well-defined. 
We also remark that both $F$ and $P \pi( q(F) )P$ are lifts in $\End_B(\ell^2(\Gamma, B) )$ of $q(F) \in \big(C_0(\partial\Omega, B), \rtimes \Gamma\big)^{\sim}$,
where $P \pi( q(F) )P$ comes from Proposition \ref{prop:Busby_of_interface_extension}. It follows 
that $[F] = \big[ P \pi(q(F)) P \big] \in KK(\C, B)$ as $F - P\pi(q(F))P$ is a compact 
operator on $\ell^2(\Gamma, B)$ (see Remark \ref{Remark_KKgroup_notes}).

\subsubsection{Free-fermionic symmetries and the interface index}

The assumption of gapped bulk systems is quite common in the study of topological phases. 
Bulk systems in condensed matter physics may also possess   free-fermionic symmetries, whose possible topological 
phases exhaust the real and complex $K$-theory groups of the $C^*$-algebra $A_\mathrm{bulk}$ modeling the  bulk  system.
These symmetries and topological phases can be described by van Daele $K$-theory 
of $\Z_2$-graded $C^*$-algebras (possibly with a real structure). We will omit a detailed description of this construction 
and refer the reader to~\cite{Kellendonk15, AMZ}. 

Following the perspective of~\cite{AMZ}, the bulk topological phase is described by a skew-adjoint unitary that anti-commutes with 
\emph{pseudo}-symmetries $\{\kappa_j\}_{j=1}^n$ related to the free-fermionic symmetry. These pseudo-symmetries generate 
a Clifford algebra $\Cl_{0,n}$ and the corresponding topological phase is described by van Daele $K$-theory,  $DK(A_\mathrm{bulk} \otimes \Cl_{0,n+1})$. 
The index $n \in \N$ depends on the symmetry under consideration.

We can also consider an interface index with free-fermionic symmetries by extending the pseudo-symmetries to operators on $\ell^2(\Gamma, B)$.

\begin{prop}
Let $\ell^2(\Gamma, B)$ be an interface with free-fermionic symmetries described by the pseudo-symmetries 
$\{\kappa_j\}_{j=1}^n \subset \End_B(\ell^2(\Gamma, B))$. If $F \in \End_B( \ell^2(\Gamma, B))$ is a skew-adjoint Fredholm 
operator that anti-commutes with $\{\kappa_j\}_{j=1}^n$, then the symmetric interface index 
$[F] \in KK(\C, B \otimes \Cl_{0,n+1})  \cong DK(B \otimes \Cl_{0,n}) \cong K_{n+1}(B)$ is well-defined.
\end{prop}
\begin{proof}
We take the triple 
\[
   \Big( \C, \, \big(\ell^2(\Gamma, B) \otimes \C^2\big)_{B\otimes \Cl_{0,n+1}}, \, \begin{pmatrix} 0 & -F \\ F & 0 \end{pmatrix} \Big)
\]
where $\ell^2(\Gamma, B) \otimes \C^2$ is graded by $\left( \begin{smallmatrix} \one & 0 \\ 0 & -\one \end{smallmatrix}\right)$. 
The Clifford action is generated by the operators 
\[
   \begin{pmatrix} 0 & -\one \\ \one & 0 \end{pmatrix},  \begin{pmatrix} 0 & \kappa_1 \\ \kappa_1 & 0 \end{pmatrix}, \ldots,  \begin{pmatrix} 0 & \kappa_n \\ \kappa_n & 0 \end{pmatrix},
\]
which are skew-adjoint (odd) unitiares that mutually anti-commute with each other. The operator  
$\left( \begin{smallmatrix} 0 & -F \\ F & 0 \end{smallmatrix}\right) \in \End_{B\otimes \Cl_{0,n+1}}(\ell^2(\Gamma, B) \otimes \C^2)$ 
is self-adjoint, odd, Fredholm and \emph{commutes} with the Clifford generators. Hence we have a well-defined Kasparov module and 
equivalence class $[F] \in KK(\C, B \otimes \Cl_{0,n+1})$.
\end{proof}

\begin{example}[Case $B=M_N(\C)$ and Hilbert spaces]
Suppose $B = M_N(\C)$ for some $N$ and so our interface $C^*$-module is a Hilbert space $\ell^2(\Gamma, \C^{N} )$. 
Then any  Fredholm operator $F \in \calB( \ell^2(\Gamma, \C^{N} ))$ has an integer-valued index 
\[
    \Index(F) = \mathrm{dim} \Ker(F) - \mathrm{dim} \Ker(F^*) \in \Z,
\]
which is the image of $[F] \in KK(\C, \C)$ under the isomorphism $KK(\C, \C) \cong \Z$. 
If $F$ is skew-adjoint and anti-commutes with pseudo-symmetries $\{\kappa_j\}_{j=1}^n$, then 
$\Ker(F)$ defines a $\Cl_{0,n}$-module. The Atiyah--Bott--Shapiro construction provides 
a $\Z$, $\Z_2$ or trivially-valued index (depending on the degree $n$), which represents the image 
of the symmetric interface index $[F] \in KK(\C, \C \otimes \Cl_{0,n+1}) \cong K_{n+1}(\R)$~\cite{AS69}. 
Here the free-fermionic symmetries also introduce a real structure and take us to real $K$-theory.
\end{example}

Our interface index also satisfies a weak version of the bulk-interface correspondence, which follows from Proposition \ref{prop:DK_boundary_to_Fredholm} in the appendix.

\begin{prop} \label{prop:weak_symm_BBC}
Let $A \cong C_0(\Omega, B)$ be a $C^*$-algebra satisfying Assumption \ref{assump:Extension_works}. Suppose  that the interface 
$\ell^2(\Gamma, B)$ has free-fermionic symmetries represnted by $\{\kappa_j\}_{j=1}^n$.  If the bulk systems are uniformly gapped as represented by an invertible 
element $w \in \big( C_0(\partial \Omega) \rtimes \Gamma )^\sim$, then the interface index 
$\big[ P \pi(w) P \big] \in KK(\C, B\otimes \Cl_{0,n+1})$ is the image of the total bulk topological phase 
$[w] \in DK\big( C_0(\partial \Omega, B) \rtimes \Gamma \otimes \Cl_{0,n+1} \big)$ under the composition 
\[
    DK\big( C_0(\partial \Omega, B) \rtimes \Gamma \otimes \Cl_{0,n+1} \big) \xrightarrow{ \delta } DK( B \otimes \Cl_{0,n} ) \xrightarrow{\simeq} KK(\C, B \otimes \Cl_{0, n+1} ).
\]
\end{prop}

For a generic interface system, there is no guarantee that the boundary map 
$DK\big( C_0(\partial \Omega, B) \rtimes \Gamma \otimes \Cl_{0,n+1} \big) \xrightarrow{ \delta } DK( B \otimes \Cl_{0,n} )$ 
is an isomorphism and generally will not be. Still, a non-trivial interface index implies that 
the total bulk class
$DK\big( C_0(\partial \Omega, B) \rtimes \Gamma \otimes \Cl_{0,n+1} \big) $ is also non-trivial. 
When we can further decompose classes in $DK\big( C_0(\partial \Omega, B) \rtimes \Gamma \otimes \Cl_{0,n+1} \big)$ 
into separate bulk systems, Proposition \ref{prop:weak_symm_BBC} says that at least some of these bulk classes 
must be non-trivial.

Proposition \ref{prop:weak_symm_BBC} also holds for the interface index without symmetry, 
where it is more natural to use complex $K$-theory. Namely, $[P \pi(w) P] = \delta[w] \in KK(\C, B)$ with 
$[w] \in K_1\big( C_0(\partial \Omega, B) \rtimes \Gamma \big)$ (see Proposition \ref{prop:K1_to_Fredholm}).

\subsection{Decomposition of the interface index via disjoint bulk systems}

We have shown that a $K$-theoretic interface index can be well-defined for a wide variety of discrete interfaces. 
On the other hand, computing the value of this index is quite challenging for a general $\Gamma$ and 
Fredholm $F \in \big( C_0(\Omega, B) \rtimes \Gamma \big)^\sim \subset \End_B( \ell^2(\Gamma, B) )$.

Recalling our basic example of the domain wall, $\Gamma = \Z$ from Section \ref{sec:DomainWall}, we used the fact that the bulk systems at 
$\pm \infty$ are spatially separated to decompose the total interface index into a difference of two indices related to the two bulk systems.
Such a decomposition will not hold in general, but as a first step we can consider the case that the system at infinity $\partial \Omega$ 
has a finite and $\Gamma$-invariant decomposition into subspaces.
Hence for this subsection, we assume there are open or closed sets $Z_1, \ldots, Z_M \subset \partial \Omega$ such that 
\begin{align} \label{eq:Boundary_decomp_assumption} 
  &\partial \Omega = \bigsqcup_{j=1}^M Z_j,   &&Z_j \cap Z_k = \emptyset, \quad j\neq k, 
  && \Gamma \cdot Z_j \subset Z_j \ \text{for all }j.
\end{align}

Interfaces satisfying this property include the functions asymptotically supported on disjoint cones from Section \ref{subsec:asymp_cone}. We can 
also build interfaces satisfying Equation \eqref{eq:Boundary_decomp_assumption} by taking disjoint, $\Gamma$-invariant and 
second countable subsets
$Z_j \subset \beta \Gamma \setminus \Gamma$ with  $\partial \Omega = \bigsqcup_j Z_j$ and using the construction from 
Section \ref{subsec:Pullback_example}. 
It is immediate from  Equation \eqref{eq:Boundary_decomp_assumption} that 
\begin{equation} \label{eq:bdry_alg_directsum}
    C_0( \partial \Omega, B) \cong  \bigoplus_{j=1}^M C_0(Z_j, B), \qquad \qquad 
    C_0( \partial \Omega, B) \rtimes \Gamma \cong  \bigoplus_{j=1}^M C_0(Z_j, B) \rtimes \Gamma.
\end{equation}

We can similarly decompose the interface short exact sequence into a sum of extensions.

\begin{prop} \label{prop:ext_disjoint_decomp}
Let $A \cong C_0(\Omega, B)$ be a $C^*$-algebra satisfying Assumption \ref{assump:Extension_works} with  
$[\tau_{\partial \Omega}] \in \Ext^{-1}( C_0(\partial \Omega, B) \rtimes \Gamma, B \otimes \calK)$ the extension class of the 
interface short exact sequence. Suppose that $\partial \Omega$ admits a $\Gamma$-invariant decomposition as in 
Equation \eqref{eq:Boundary_decomp_assumption}. Then there are extensions $[\tau_{Z_j}] \in \Ext^{-1}( C_0(Z_j, B) \rtimes \Gamma, B \otimes \calK)$ 
and a function $\mathrm{sgn}:\{1,\ldots, M\} \to \{1,-1\}$ such that 
\[
     [\tau_{\partial \Omega}] = \bigoplus_{j=1}^M [\tau_{Z_j}]^{\mathrm{sgn}(j) },
\]
where $[\tau]^{-1}$ denotes the inverse of $[\tau]$ in the group of invertible extensions.
\end{prop}

\begin{remark}
Proposition \ref{prop:ext_disjoint_decomp} is a little difficult to state precisely. We are considering the classes 
$[\tau_{Z_j}] \in \Ext^{-1}( C_0(Z_j, B) \rtimes \Gamma, B \otimes \calK)$, while 
$ [\tau_{\partial \Omega}] \in \Ext^{-1}( C_0(\partial \Omega, B) \rtimes \Gamma, B \otimes \calK)$, a different group. 
We are using that Equation \eqref{eq:bdry_alg_directsum} induces  a canonical isomorphism 
\[
   \bigoplus_{j=1}^M \Ext^{-1}( C_0(Z_j, B) \rtimes \Gamma, B \otimes \calK) 
   \ni \bigoplus_{j=1}^M [\tau_j] \xmapsto{\simeq} \Big[ \bigoplus_{j=1}^M \tau_j \Big] \in \Ext^{-1}( C_0(\partial \Omega, B) \rtimes \Gamma, B \otimes \calK)
\]
such that 
\[
   \bigoplus_{j=1}^M [\tau_{Z_j}]^{\mathrm{sgn}(j) } \mapsto [\tau_{\partial \Omega}]  \in \Ext^{-1}( C_0(\partial \Omega, B) \rtimes \Gamma, B \otimes \calK).
\]
It will be clear from the context where the appropriate extension class and inverses are being taken.
\end{remark}

\begin{proof}[Proof of Proposition \ref{prop:ext_disjoint_decomp}]
The extension $[\tau_{\partial \Omega}]$ is built from the completely positive and contractive map 
$\rho_0: C_0(\partial \Omega) \to C(\beta \Gamma)$ from Lemma \ref{lemma:scalar_cp_splitting}. We can repeat this argument for each $Z_j$ 
 to obtain completely positive and contractive  maps $\rho^{Z_j}_0: C_0(Z_j) \to C(\beta \Gamma)$ and the extension to the crossed product 
\[
   \rho^{Z_j}: C_0(Z_j, B)\rtimes \Gamma \to \End_B( \ell^2(\Gamma, B) )
\]
for each $j=1,\ldots, M$.
Similarly, we can apply Kasparov's Stinespring Theorem to obtain the projections $\{P_1,\ldots, P_M\} \subset \End_B(\ell^2(\Gamma, B)^{\oplus 2} )$ and 
representations $\pi_{Z_j} : C_0( Z_j, B) \rtimes \Gamma \to \End_B(\ell^2(\Gamma, B)^{\oplus 2} )$ such that $P_j \pi_{Z_j}(a) P_j = \rho^{Z_j}(a) \oplus 0$ 
for any $a \in C_0(Z_j, B)\rtimes \Gamma$ and $j=1,\ldots, M$. 
It is clear that the semi-splitting $\rho_0: C_0(\partial \Omega) \to C(\beta \Gamma)$ is such that 
$\rho_0 = \bigoplus_j \rho_0^{Z_j}$ under the decomposition $C_0( \partial \Omega ) \cong  \bigoplus_{j=1}^M C_0(Z_j )$. 
However, it is not guaranteed that the translation action of $\Gamma$ on each $Z_j$ is globally consistent. In order to simultaneously consider an 
action of $\Gamma$ on $\bigsqcup_j Z_j$, we may need to reverse the orientation for some of the subsets. In particular, we can fix a function 
$\mathrm{sgn}:\{1,\ldots, M\} \to \{1,-1\}$ such that for all 
$a = (a_1,\ldots, a_M) \in C_0(\partial \Omega, B) \rtimes \Gamma \cong \bigoplus_{j=1}^M C_0(Z_j, B) \rtimes \Gamma$,
\[
   \bigoplus_{j=1}^M P_j^{\mathrm{sgn}(j)} \pi_{Z_j}(a_j) P_j^{\mathrm{sgn}(j)}  = P\pi(a) P, \qquad 
   P_j^{\mathrm{sgn}(j)} = \begin{cases} P_j, &  \mathrm{sgn}(j) =1, \\ \one- P_j, & \mathrm{sgn}(j)=-1 \end{cases}.
\]
Because the $\ast$-homomorphism $\tau_{\one-P_j}(a_j) = q\big( (\one-P_j) \pi(a_j) (\one- P_j) \big)$, $a_j \in C_0(Z_j, B)\rtimes \Gamma$, will represent the inverse of $\tau_{P_j}$, we see that 
\begin{align*}
   [\tau_{\partial \Omega}] &= \big[ q\big( P \pi( \, \cdot \, ) P \big) \big] = \bigoplus_{j=1}^M  \Big[ q \big( P_j^{\mathrm{sgn}(j)} \pi_{Z_j}(\, \cdot \, ) P_j^{\mathrm{sgn}(j)} \big) \Big] \\
   &=     \bigoplus_{j=1}^M [\tau_{Z_j}]^{\mathrm{sgn}(j) } \in \Ext^{-1}( C_0(\partial \Omega, B) \rtimes \Gamma, B \otimes \calK)
\end{align*}
as required.
\end{proof}

For any $T \in C_0(\Omega, B) \rtimes \Gamma$, we can decompose 
\[
    q(T) = \big( q_1(T), \ldots, q_M(T) \big)\in C_0( \partial \Omega, B) \rtimes \Gamma , \qquad  
    q_j(T) \in C_0(Z_j, B) \rtimes \Gamma.
\]
By Proposition \ref{prop:ext_disjoint_decomp}, we can build an extension for each component of the decomposition $\partial \Omega \cong \bigsqcup_j Z_j$, 
\[
  0 \to \calK \otimes B \to C^*\big( P_j \pi_{Z_j}(a) P_j, \, K \, \big| \, a \in C_0( Z_j, B) \rtimes \Gamma, \, K \in \calK \otimes B \big) \to C_0( Z_j, B) \rtimes \Gamma \to 0
\]
for all $j=1,\ldots, M$. Splitting up the interface extension allows us to decompose the interface index.

\begin{thm} \label{thm:Interface_index_decomp}
Let $A \cong C_0(\Omega, B)$ be a $C^*$-algebra satisfying Assumption \ref{assump:Extension_works} and $F \in \big(C( \Omega, B) \rtimes \Gamma\big)^\sim$ 
a Fredholm operator on $\ell^2(\Gamma, B)$. If $\partial \Omega$ admits a $\Gamma$-invariant decomposition as in 
Equation \eqref{eq:Boundary_decomp_assumption}, then there is  a function $\mathrm{sgn}:\{1,\ldots, M\} \to \{1,-1\}$ such that 
 the interface index 
\[
   [F] = \bigoplus_{j=1}^M  [F_j]^{\mathrm{sgn}(j)} \in KK(\C, B) \cong K_0(B),
\]
where $F_j$ is a Fredholm operator on a submodule of $\ell^2(\Gamma, B)^{\oplus 2}$ for all $j=1,\ldots, M$.

If the interface $\ell^2(\Gamma, B)$ has free-fermionic symmetries described by $\{\kappa_i\}_{i=1}^n$, then 
\[
   [F] = \bigoplus_{j=1}^M  [F_j]^{\mathrm{sgn}(j)} \in KK(\C, B\otimes \Cl_{0,n+1}) \cong DK( B \otimes \Cl_{0,n} ) \cong K_{n+1}(B).
\]
\end{thm}
\begin{proof}
Let $q(F) \in \big(C_0(\partial \Omega, B)\rtimes \Gamma\big)^\sim$ be the invertible image of $F$ under the quotient by 
$\mathbb{K}_B(\ell^2(\Gamma, B))$. Then 
by Proposition \ref{prop:Busby_of_interface_extension}, $F - P \pi( q(F) ) P \in \mathbb{K}_B(\ell^2(\Gamma, B))$ and so 
\[
  [F] = \big[ P \pi( q(F) ) P \big] = \bigoplus_{j=1}^M \big[ P_j \pi_{Z_j} ( q_j(F) ) P_j \big]^{\mathrm{sgn}(j) },
\]
where the last equality is due to Proposition \ref{prop:ext_disjoint_decomp}. We let 
$F_j =  P_j^{\mathrm{sgn}(j)} \pi_{Z_j} ( q_j(F) ) P_j^{\mathrm{sgn}(j)}$, which is Fredholm on 
$P_j^{\mathrm{sgn}(j)} \cdot \ell^2( \Gamma, B)^{\oplus 2} \cong \ell^2(\Gamma, B)$ and the result follows. 
The same argument holds for the interface index with free-fermionic symmetries.
\end{proof}

If we consider the sets $\{ Z_j\}_{j=1}^M$ as representing the spatial location of the separated bulk systems at infinity, Theorem \ref{thm:Interface_index_decomp} says that 
the interface index decomposes into a sum, where each $[F_j] \in KK(\C, B)$ depends on a single bulk system. We expect these refined indices to be easier to 
compute in practice than the entire interface index. 

For general systems, we can take a decomposition of $\partial \Omega$ into quasi-orbits, $\{\Xi_j\}_{j\in J}$, and there is an injective $\ast$-homomorphism 
\[
   C_0(\partial \Omega, B) \rtimes \Gamma \to \prod_{j\in J} C_0( \Xi_j, B) \rtimes \Gamma.
\]
If $w_j \in \big(C_0( \Xi_j, B) \rtimes \Gamma \big)^{\sim}$ is invertible for some $j \in J$, then there is a well-defined index 
\[
  \big[ F_{\Xi_j} \big] = \big[ P_j \pi_{\Xi_j}(w_j) P_j ] \in KK(\C, B).
\]
If $w_j \in \big(C_0( \Xi_j, B) \rtimes \Gamma \big)^{\sim}$ is invertible for all $j \in J$, 
  we do not expect a  decomposition of $[F]$ into a signed sum of $[F_{\Xi_j}]$ to hold in
  general when $\Xi_j\cap \Xi_k \neq \emptyset$ (even if the index set $J$ is finite). 
A more thorough study of how one can decompose the interface index in cases where $\partial \Omega$ naturally decomposes into 
subspaces with non-empty intersection (such as the example of Cartesian anisotropy from Section \ref{subsec:Cartesian_Aniso}) would certainly 
be desirable. We leave this question to future work.

\appendix

\section{Review of Hilbert $C^*$-modules and Kasparov theory} \label{sec:Kasparov_review}

For the benefit of the reader, we give a short overview of the basics of 
Hilbert $C^*$-modules and noncommutative index theory. We refer to the textbooks~\cite{Blackadar, Lance, WO} 
for a more detailed introduction and proofs.

\subsection{Hilbert $C^*$-modules} \label{subsec:C*_modules}

Throughout this section, $B$ is a $\sigma$-unital $C^*$-algebra. If $B$ is separable, this condition 
is equivalent to the existence of a countable approximate unit.

\begin{defn}
A Hilbert $C^*$-module is a right $B$-module $X_B$ with a map
$(\cdot\mid\cdot)_B: X_B \times X_B \to B$ that is linear in the second  variable such that 
for all $x_1, x_2 \in X_B$ and $b \in B$, 
\[
  	(x_1 \mid x_2\cdot b)_B = (x_1\mid x_2)_B \, b, \qquad \qquad 
  	( x_1\mid x_2)_B = (x_2 \mid x_1)_B^*, \qquad \qquad
  	(x_1 \mid x_1)_B \geq 0
\]
and $(x_1 \mid x_1)_B = 0$ if and only if $x_1=0$. Furthermore, $X_B$ is complete with 
respect to the norm $\|x \| = \big\| ( x \mid x)_B \big\|_B^{1/2}$, $x \in X_B$.
\end{defn}

\begin{examples}
\begin{enumerate}
  \item[(i)] A  complex Hilbert space $\calH$ is a Hilbert $C^*$-module over $\C$ with $(\cdot\mid \cdot)_\C$ given by the 
  usual Hilbert space inner product.
  
  \item[(ii)] The algebra $B$ can be seen as a Hilbert $C^*$-module $B_B$ with the structure 
  \[
     b_1 \cdot b_2 = b_1 b_2, \qquad \qquad (b_1 \mid b_2)_B = b_1^* b_2, \qquad b_1, b_2 \in B.
  \]
  
  \item[(iii)] Given a separable and infinite dimensional Hilbert space $\calH$, the  standard Hilbert $C^*$-module 
  $\calH_B := \calH \otimes B$ is given the $C^*$-module structure 
  \[
      (\psi \otimes b_1) \cdot b_2 = \psi \otimes b_1 b_2, \qquad 
      \big( \psi_1 \otimes b_1 \mid \psi_2 \otimes b_2 \big)_B = \langle \psi_1 , \psi_2 \rangle_\calH \, b_1^* b_2.
  \]
  Kasparov's Stabilisation Theorem implies that for any countably generated Hilbert $C^*$-module $X_B$, there is an isometric embedding 
 $X_B \to \calH_B$~\cite[Theorem 2]{KasparovStinespring}.
\end{enumerate}
\end{examples}

A Hilbert $C^*$-module $X_B$ is full if $\ol{\{ (x_1 \mid x_2)_B \,:\, x_1, x_2 \in X_B \}} = B$ (closure in $B$).   
Both $B_B$ and $\calH_B$ are full Hilbert $C^*$-modules.

\begin{remark}
Hilbert $C^*$-modules share many similarities to Hilbert spaces. One important difference  is that if $Y_B$ is a closed submodule 
of $X_B$ with 
\[
   Y^\perp_B = \big\{ x \in X_B \,: \, (x\mid y)_B = 0 \text{ for all } y \in Y_B \big\},
\]
then it is not guaranteed that $X_B = Y_B \oplus Y^\perp_B$ in general.
\end{remark}

Given a map $T: X_B \to X_B$, we say that $T$ is adjointable if there is a map $T^*:X_B \to X_B$ that acts as the adjoint 
with respect to the $B$-valued inner product on $X_B$. Such maps are right $B$-linear and bounded in the operator norm. 
The set of adjointable operators on $X_B$ is denoted $\End_B(X)$, which is a $C^*$-algebra with respect to the operator norm. 
In analogy to finite-rank operators on Hilbert spaces, 
for any $x_1, x_2 \in X_B$ we define  $\Theta_{x_1,x_2}(x_3) = x_1 \cdot (x_2\mid x_3)_B$, $x_3 \in X_B$. 
One finds that  $\Theta_{x_1,x_2}^* = \Theta_{x_2,x_1}$, so $\Theta_{x_1,x_2}$ is adjointable. The 
compact operators $\mathbb{K}_B(X)$ on $X_B$ are defined such that 
\[
    \mathbb{K}_B(X) = \ol{ \mathrm{span}\big\{ \Theta_{x_1,x_2} \, :\, x_1 , x_2 \in X_B \big\} } \subset \End_B(X).
\]
The compact operators are a closed two-sided ideal in $\End_B(X)$ and we denote the 
$C^*$-algebra $\calQ_B(X) = \End_B(X)/ \mathbb{K}_B(X)$. 
Note that $\calQ_B(X)$ is a 
generalisation of the Calkin algebra to the setting of Hilbert $C^*$-modules.

\begin{examples}
\begin{enumerate}
  \item[(i)] Considering a Hilbert space $\calH_\C$ as a Hilbert $C^*$-module over $\C$, the compact and adjointable operators on $\calH_\C$ are 
  precisely the compact and bounded operators on $\calH$ respectively.
  
  \item[(ii)] Considering $B$ as a Hilbert $C^*$-module over itself, $\mathbb{K}_B(B) = B$ as multiplication is dense in $B$. 
  The adjointable operators are isomorphic to the multiplier algebra of $B$, $\Mult(B) \cong \End_B(B)$.
  
  \item[(iii)] For the standard $C^*$-module $\calH_B$,   $\mathbb{K}_B(\calH_B) = \calK(\calH) \otimes B$ and 
 $\End_B( \calH_B) \cong \Mult( \calK(\calH) \otimes B)$.
\end{enumerate}
\end{examples}

\subsection{Fredholm operators on  Hilbert $C^*$-modules and $K$-theory} \label{sec:KK}

\begin{defn}
We say that an operator $T \in \End_B(X)$ is Fredholm on $X_B$ if $q(T) \in \calQ_B(X)$ is invertible. 
\end{defn}

\begin{lemma}[{\cite[Lemma 2.7]{Wahl07}}] \label{lem:unbdd_Fred_condition}
A self-adjoint element $T \in \End_B(X)$ with $\| T\| \leq 1$ is Fredholm  if and only if $\| \one - q(T)^2 \|_{\calQ_B(X)} < 1$.
\end{lemma}

Self-adjoint Fredholm operators on $C^*$-modules can be used to give a description of $K$-theory of $C^*$-algebras as a special 
case of Kasparov's $KK$-theory. We give a very brief introduction sufficient for the contents of this paper and refer the reader 
to~\cite{Kasparov80, Blackadar} for the complete theory.

We consider $C^*$-modules $X_B$ with a $\Z_2$-grading, $X_B \cong (X^0 \oplus X^1)_B$. 
Occasionally we consider the case where $B$ is also $\Z_2$-graded, $B= B^0 \oplus B^1$,  where the right-action 
is such that $X^i \cdot B^j \subset X^{i+j \, \mathrm{mod} \, 2}$. The grading on $X_B$ also extends to the adjointable 
operators, where 
$T \in \End_B(X)$ is even (resp. odd) if $T\cdot X^j \subset X^j$ (resp. $T\cdot X^j \subset X^{j+1 \, \mathrm{mod} \, 2}$).

\begin{defn}
Let $X_B$ be a countably generated and $\Z_2$-graded Hilbert $C^*$-module and $F \in \End_A(E)$ a self-adjoint odd operator such that 
$\one - F^2 \in \mathbb{K}_B(X)$. We call the triple $(\C, X_B, F)$ a Kasparov module.
\end{defn}

\begin{example} \label{ex:BasicKasMod}
Let $X_B$ be an ungraded $C^*$-module and $F \in \End_B(X)$ such that $\one - FF^*, \one - F^*F \in \mathbb{K}_B(X)$. We consider the 
doubled $C^*$-module $X_B \otimes \C^2$ with a $\Z_2$-grading via the operator $\left( \begin{smallmatrix} \one & 0 \\ 0 & -\one \end{smallmatrix} \right)$. 
Then the triple 
\[
   \Big( \C, \, X_B \otimes \C^2 , \, \begin{pmatrix} 0 & F^* \\ F & 0 \end{pmatrix} \Big) 
\]
is a Kasparov module. 

For a generic Fredholm operator $F \in \End_B(X)$, we can build a Kasparov module by taking $X_B \otimes \C^2$ and  
$\tilde{F} = \chi \left( \begin{smallmatrix} 0 & F^* \\ F & 0 \end{smallmatrix} \right)$ with $\chi: \R \to \R$ a smooth odd function such that 
$\chi'(0)>0$, $\lim_{x\to \infty} \chi(x) = 1$ and $\one - \tilde{F}^2 \in \mathbb{K}_B( X_B \otimes \C^2)$. Such a function always 
exists by~\cite[Corollary 2.2]{Wahl07}. Then 
\[
   \big( \C, \, X_B \otimes \C^2, \, \tilde{F} \big)
\]
is a Kasparov module.
\end{example}

Two Kasparov modules $(\C, X^{(0)}_B, F_0)$ and $(\C, X^{(1)}_B, F_1)$ are unitarily equivalent if there is an 
even unitary  $U: X^{(1)}_B \to X_B^{(2)}$ 
such that $UF_1 U^* = F_2$. 
We say that  the Kasparov modules are homotopic if there is an   Kasparov module 
$(\C, \tilde{X}_{B \otimes C([0,1])}, F)$ such that the evaluating the fibre at $0$ and $1$ yields Kasparov modules 
that are unitarily equivalent to $(\C, X^{(0)}_B, F_0)$ and $(\C, X^{(1)}_B, F_1)$ respectively.

Homotopy equivalence classes of Kasparov modules yields an abelian group $KK(\C, B)$, 
where the  group operation is by direct sum~\cite[Section 4]{Kasparov80}.

\begin{remarks} \label{Remark_KKgroup_notes}
\begin{enumerate}
 \item If $F \in \End_B(X)$ is a self-adjoint Fredholm operator such that $F^2 = \one$, then the Kasparov module 
 $\big( \C, X_B, F \big)$ is called degenerate. Degenerate Kasparov modules represent the group identity in 
 $KK(\C, B)$.
   \item For a generic Fredholm operator $F \in \End_B(X)$, Example \ref{ex:BasicKasMod} shows that 
  we can build a Kasparov module via a normalising function 
  $\tilde{F} = \chi \left( \begin{smallmatrix} 0 & F^* \\ F & 0 \end{smallmatrix} \right)$ on $X_B \otimes \C^2$. 
  The equivalence class of this Kasparov module in $KK(\C, B)$ is independent of the choice of normalising funciton~\cite[Proposition 2.14]{vdDungen17}. 
  As such, we simply write $[F] \in KK(\C, B)$ to denote this equivalence class.
  \item  If $F_0, F_1 \in \End_B(X)$ are   Fredholm operators such that $F_0 - F_1 \in \mathbb{K}_B(X)$, then 
$[F_0] = [F_1] \in KK(\C, B)$, which can be shown by the straight-line homotopy 
$[0,1] \ni t  \mapsto F_t = F_0 + t(F_1 - F_0)$. 
\end{enumerate}
\end{remarks}

When $B$ is a trivially graded  $C^*$-algebra $(B^1 = \{0\}$), there is an isomorphism $KK(\C, B) \cong K_0(B)$, the operator 
algebraic $K$-theory of $B$ constructed from equivalence classes of projections in $\bigoplus_n M_n(B)$~\cite[Section 6]{Kasparov80}. 
When $B=\C$,  the equivalence $KK(\C, \C) \cong K_0(\C) \cong \Z$ is given by the Fredholm index 
\[
   KK(\C, \C) \ni \Big[ \Big( \C, \, \calH \otimes \C^2, \, \begin{pmatrix} 0 & F^* \\ F & 0 \end{pmatrix} \Big) \Big] 
   \xmapsto{\simeq} \Index(F) = \mathrm{dim} \Ker(F) - \mathrm{dim} \Ker(F^*) \in \Z.
\]

When $B$ is $\Z_2$-graded, one can consider the van Daele $K$-theory group $DK(B)$, which is constructed from 
equivalence classes of odd-self adjoint unitaries in $\bigoplus_n M_n(B)$~\cite{vanDaele1}. 
We also have an isomorphism $KK(\C, B) \cong DK( B \,\hat\otimes\, \Cl_{1,0} )$, with $\Cl_{1,0}$ the Clifford 
algebra with one self-adjoint (odd) generator of square $+1$~\cite{Roe04}.
By incorporating a real structure to $C^*$-algebras and $C^*$-modules, we can also describe all real $K$-theory groups. 
In particular, if $B$ is ungraded, there are isomorphisms 
\[
    KK(\C, B \otimes \Cl_{r,s} ) \cong DK( B \otimes \Cl_{r+1,s} ) \cong K_{s-r}(B) ,
\]
see~\cite{BKR2, Kubota16} for further details.

%%%%%%%%%%%%%%%%%%%%%%%%%%%%%%%%%%%%%%%%%%%%%%%%%%%%%%%%%%%%%%%%%%%%%%%%%%%%%%%%%%%%%%%%%%%%%%%
\section{Extension theory of $C^*$-algebras} \label{sec:extensions}
%%%%%%%%%%%%%%%%%%%%%%%%%%%%%%%%%%%%%%%%%%%%%%%%%%%%%%%%%%%%%%%%%%%%%%%%%%%%%%%%%%%%%%%%%%%%%%%

In this section, we provide an overview of some of the main results from extension theory of 
$C^*$-algebras and its connection to Fredholm operators and $K$-theory. Further details can be 
found in~\cite[Chapter VII]{Blackadar} and \cite[Chapter 3]{WO}. 

\subsection{Basic definitions and properties}

\begin{defn}
Let $A$ and $B$ be $C^*$-algebras. An \emph{extension of $A$ by $B$} is a 
short exact sequence of $C^*$-algebras
\begin{equation}\label{eq_ses}
 0 \to B \to E \xrightarrow{q} A \to 0.
\end{equation}
Two extensions of $A$ by $B$ are \emph{strongly isomorphic} if there is an 
isomorphism $\eta: E_1 \to E_2$ such that the following diagram commutes,
\[
\xymatrix{
    0 \ar[r] & B \ar[r] \ar@{=}[d] & E_1 \ar[r] \ar[d]_\eta & A \ar[r] \ar@{=}[d] & 0 \\
    0 \ar[r] & B \ar[r] & E_2 \ar[r] & A \ar[r] & 0.
}
\]
The extension from Equation \eqref{eq_ses} is \emph{trivial} if there is a $\ast$-homomorphism $s:A \to E$ such that 
$q \circ s = \mathrm{Id}_A$.
\end{defn}

Note that trivial extensions are those in which the exact sequence splits. 
In this case, one has $E \cong A \oplus B$ \cite[Proposition 3.1.3]{WO}.

Given $A$ and $B$, one can try and classify all extensions of $A$ 
by $B$ up to strong isomorphism (or another notion of equivalence). Such a question is closely 
tied to the  Busby invariant. Given the extension \eqref{eq_ses}, 
because $B$ is an ideal of $E$, there is a homomorphism $\nu: E \to \Mult(B)$, 
the multiplier algebra of $B$. Composing $\nu$ 
with the quotient map $q_B: \Mult(B) \to \Mult(B) / B =: \calQ(B)$ yields the Busby invariant.

\begin{defn}
For the extension 
$0 \to B \to E \to A \to 0$ of $A$ by $B$, 
the homomorphism $\tau: A \to \calQ(B)$ defined by $\tau(a) := q_B \circ \nu(e)$, 
for any lift $e\in E$ of $a \in A$, is called the \emph{Busby invariant} of the extension. 
\end{defn}

The Busby invariant $\tau:A \to Q(B)$ is injective if and only if 
$B$ is an essential ideal in $E$,  $e  B=0$ implies $e=0$ for $e \in E$. 
We will only be interested in extensions 
of  $A$ by  $B$ with $\tau$ injective.

\begin{example} \label{ex:Extension_by_compacts}
Suppose that $\calH$ is an infinite-dimensional Hilbert space and let 
$B = \calK(\calH)$. Then $\Mult(B) \cong \calB(\calH)$ and 
$\calQ(B) = \calB(\calH)/\calK(\calH) = \calQ(\calH)$ is the standard Calkin algebra. 
This means that the Busby invariant of an extension of $A$ by $\calK(\calH)$ 
is a $\ast$-homomorphism $\tau: A \to \calQ(\calH)$.
\end{example}

Given an extension, the algebra $E$ is isomorphic to the pullback $C^*$-algebra of 
$\big(A, \Mult(B)\big)$ along $(\tau, q_B)$, i.e. the algebra $\calP$ such that 
\[
 \xymatrix{
   \calP \ar[r] \ar[d]  & \Mult(B) \ar[d]^{q_B} \\ A \ar[r]^{\tau} & \calQ(B)
 }
\]
commutes. That is, $\calP = \big\{(a, m) \in A \oplus \Mult(B) \mid \tau(a) = q_B(m) \big\}$. 
Thus, we have that two extensions are strongly isomorphic if and only if their Busby invariants coincide. 
Similarly, an extension is trivial if and only if $\tau: A \to \calQ(B)$ lifts to 
a homomorphism $A \to \Mult(B)$. We can therefore 
describe an extension purely in terms of the homomorphism $\tau:A \to \calQ(B)$.

A reason for introducing the Busby invariant is that it allows us to compare different extensions of 
$A$ by $B$. In particular, it may be difficult to directly compare $E_1$ and 
$E_2$ such that 
$0 \to B \to E_1 \to A \to 0$ and $0 \to B \to E_2 \to A \to 0$ are exact. 
But their corresponding Busby invariants $\tau_1$ and $\tau_2$ are $\ast$-homomorphisms from $A$ to the same algebra $\calQ(B)$.
We say that two extensions of $A$ by $B$ with Busby invariants $\tau_1$ and $\tau_2$ are strongly equivalent 
if there is a unitary $u \in \Mult(B)$ such that $\tau_2(a) = q_B(u) \tau_1(a) q_B(u)^*$ for all $a \in A$.

To define further structure on extensions, we now replace $B$ with $B\otimes \calK$ 
with $\calK = \calK(\calH)$ the compact operators
for some separable and infinite dimensional Hilbert space $\calH$. 
In particular, fixing an isomorphism $\calK \cong M_2(\calK)$ will set isomorphisms
\begin{equation}\label{eq_stand_iso}
   M_2 \big( \Mult(B\otimes \calK) \big) \cong \Mult(B\otimes \calK), \qquad  
   M_2 \big( \calQ(B\otimes \calK) \big) \cong \calQ(B\otimes \calK), 
\end{equation}
see \cite[Definition 3.3.3]{WO} for  example.
We can then define the sum of two extensions of $A$ by $B\otimes \calK$ with Busby invariants 
$\tau_1$ and $\tau_2$ as the extension whose Busby invariant is 
\[
  \tau_1 \oplus \tau_2: A \to \calQ(B\otimes \calK) \oplus \calQ(B\otimes \calK) 
  \subset M_2 \big( \calQ(B\otimes \calK) \big) \cong \calQ(B\otimes \calK).
\]
Such an addition gives a direct sum operation on the set of strong equivalence classes of extensions. The direct sum 
is commutative and we obtain an abelian semigroup.

\begin{defn}
The extension semigroup $\Ext(A, B\otimes\calK)$ is the quotient of the set of strong equivalence classes 
of extensions of $A$ by $B\otimes\calK$ by the sub-semigroup of trivial extensions.
\end{defn}

\begin{example} \label{ex:Generalised_Toeplitz2}
Let $B$ be a $\sigma$-unital $C^*$-algebra,  $X_B$ be a countably generated $C^*$-module 
and $\pi: A \to \End_B(X)$ a representation of $A$ on $X_B$. 
We suppose there is a projection $P = P^* = P^2 \in \End_B(X) \setminus \mathbb{K}_B(X)$
such that 
\[
   [ P, \pi(a) ] \in \mathbb{K}_B(X) \ \text{for all } a \in A, \qquad 
   P \pi(a) P \in \mathbb{K}_B(X) \ \text{if and only if } a = 0.
\]
From the data $(\pi,A,  X_B, P)$, we can construct an extension via the subalgebra 
of $\End_B(X)$ generated by $P \pi(a) P$ and $\mathbb{K}_B(X)$. Namely, 
\[
   0 \to \mathbb{K}_B(X) \to C^*\big( P\pi(a) P,  K  \, \big| \,  a \in A, K\in \mathbb{K}_B(X) \big) \to A \to 0
\]
is a short exact sequence with Busby invariant $[\tau] \in \Ext( A, \mathbb{K}_B(X))$ such that 
$\tau: A \to \calQ( \mathbb{K}_B(X)) \cong \calQ_B(X)$ is given by 
$\tau(a) = q( P \pi(a) P)$. 

We can also define a Busby invariant with range in $\calQ(B \otimes \calK)$ via Kasparov's stabilisation 
theorem~\cite[Theorem 2]{KasparovStinespring}. 
Because $X_B$ is countably generated, there is an adjointable 
isometry $V: X_B \to \calH_B$ with $\calH_B = \calH \otimes B$ the standard $C^*$-module. We can then define 
$\tilde{\tau} : A \to \calQ(B \otimes \calK)$ such that $\tilde{\tau}(a) = q( VP  \pi(a) PV^* )$, which will give 
an equivalence class $[\tilde{\tau}] \in \Ext(A, B \otimes \calK)$.
\end{example}

\subsection{Semi-split extensions are invertible} \label{subsec:semi-split_inv_ext}

We say that an extension $\tau: A \to \calQ(B \otimes \calK)$ is invertible if there is another 
extension $\tau' : A \to \calQ(B \otimes \calK)$ such that 
$[\tau] \oplus [ \tau'] = [ \tau_\mathrm{triv.} ] \in \Ext(A, B\otimes \calK)$. We denote by  $\Ext^{-1}(A, B) \subset \Ext(A,  B\otimes \calK)$ 
the abelian subgroup of invertible 
extensions. One also generally writes $\Ext^{-1}(A)$ to denote 
$\Ext^{-1}(A, \calK(\calH))$. Much more is known about 
$\Ext^{-1}(A,B \otimes \calK)$ than $\Ext(A,  B\otimes \calK)$. 
For example, it was shown by Kasparov~\cite[Section 7]{Kasparov80} that equivalence classes in 
$\Ext^{-1}( A, B\otimes \calK)$ satisfy a notion of homotopy invariance 
in the case that $ A$ is separable
and $B$ is $\sigma$-unital. 

Here we review a sufficient condition for an extension to be invertible. Namely, the existence of a completely positive semi-splitting.

\begin{defn}
We say that an extension $0 \to B \to E \xrightarrow{q} A \to 0$ is semi-split if there is a completely positive and contractive map 
$\rho: A \to E$ such that $q \circ \rho = \mathrm{Id}_A$. 
\end{defn}

We note that the map $\rho: A \to E$ is a $\ast$-homomorphism modulo the ideal $B$, 
\[
    q( \rho(a_1a_2) )= a_1 a_2  = q(\rho(a_1)) q(\rho(a_2)) = q( \rho(a_1) \rho(a_2) )
\]
and so $\rho(a_1) \rho(a_2) = \rho(a_1a_2) + b$ for some $b \in B$. The extension will be trivial precisely when 
$\rho$ is a $\ast$-homomorphism.

Semi-split extensions provide an inverse to the construction in Example \ref{ex:Generalised_Toeplitz2}. We first replace 
the algebra $B$ with $B\otimes \calK$ and consider a semi-split extension 
\[
   0 \to B \otimes \calK \to E \to A \to 0
\]
such that $B \otimes \calK$ is an essential ideal in $E$. Then we can consider 
$E \subset \Mult(B \otimes \calK) \cong \End_B( \calH_B)$ and 
$\rho: A \to \Mult(B \otimes \calK)$.
We then apply 
Kasparov's generalised Stinespring theorem, which says that any completely positive map 
$\rho$ is the compression of a $\ast$-homomorphism.

\begin{thm}[{Kasparov's Generalised Stinespring Theorem, \cite[Theorem 3]{KasparovStinespring}}]
Let $A$ be a separable $C^*$-algebra and  $B$ a $\sigma$-unital $C^*$-algebra. 
Assume that $\rho: A \to \Mult( B\otimes\calK)$ is a completely positive map such that $\|\rho\| \leq 1$.  
Then there exists a $\ast$-homomorphism $\pi :A \to \Mult\big( M_2( B\otimes\calK) \big) \cong \End_B\big( \calH_B^{\oplus 2} \big)$ 
such that for all $a \in A$,  
$\rho(a) \oplus 0 = P \pi(a) P$, where $P: \End_B\big( \calH_B \oplus \calH_B \big) \to \End_B(\calH_B)$ is the 
projection onto the first summand.
If $A$ is unital and $\rho$ a unital map, then $\pi$ is unital.
\end{thm}

\begin{cor} \label{cor:compact_commutators_with_Kasparov_proj}
The representation $\pi: A \to \End_B\big( \calH_B^{\oplus 2} \big)$ is such that for all $a \in A$, 
$[P, \pi(a)] \in \mathbb{K}_B\big( \calH_B^{\oplus 2} \big) \cong M_2( B\otimes\calK)$. 
\end{cor}
\begin{proof}
For any $a_1, a_2 \in A$, we have that $\rho(a_1) \rho(a_2) = \rho(a_1a_2) + b$ with $b \in B\otimes \calK$, which implies 
\[
    P\pi(a_1) P \pi(a_1) P - P \pi(a_1 a_2)   P  =  P [\pi(a_1), P] \pi(a_2) P \in M_2( B \otimes \calK)
\]
If $\pi(a) \in M_2( B \otimes \calK)$, then $P\pi(a) P = \rho(a) \oplus 0 \in M_2(B \otimes \calK)$ and so 
$q \circ \rho( a) = 0 = a$. Hence $\pi(a) \in M_2( B \otimes \calK)$ if and only if $a = 0$. 
Similarly, $P$ can not be  in $M_2( B \otimes \calK)$ or else $P \pi(a) P \in M_2( B \otimes \calK)$ for any 
$a \in A$ and so $q( P \pi(a) P) = 0 = q( \rho(a) \oplus 0) = a$ for all $a \in A$, a contradiction. 
So it follows that $[P, \pi(a) ] \in M_2( B \otimes \calK)$ for all $a \in A$.
\end{proof}

\begin{cor}
The map 
\[
   \tau_P: A \to \calQ(B \otimes \calK), \qquad \tau_P(a) = q\big( P \pi(a) P\big)
\]
is an injective $\ast$-homomorphism that defines an invertible element in $\Ext^{-1} \big( A, \, B \otimes \calK \big)$.
\end{cor}
\begin{proof}
Kasparov's Stinespring Theorem also lets us define another $\ast$-homomorphism 
\[
  \tau_{\one- P}: A \to \calQ(B \otimes \calK), \qquad \tau_{\one -P}(a) = q\big( (\one -P) \pi(a) (\one - P)\big).
\]
For any $a \in A$, the off-diagonal terms $P\pi(a)(\one -P), (\one-P) \pi(a)P \in \mathbb{K}_B(\calH_B\oplus \calH_B)$ 
by Corollary \ref{cor:compact_commutators_with_Kasparov_proj}. In particular, we see that 
\[
    P \pi(a) P + (\one - P) \pi(a) (\one - P) - \pi(a) \in  \mathbb{K}_B(\calH_B\oplus \calH_B)
\]
and so the sum $[\tau_P] \oplus [\tau_{\one- P}]$ can be described by the completely positive map 
$a \mapsto \pi(a)$, which is a $\ast$-homomorphism. Hence the extension splits and is trivial in $\Ext(A, B\otimes \calK)$. 
\end{proof}

Not every extension of $C^*$-algebras is semi-split. A sufficient condition is that the quotient algebra $A$ is separable and 
nuclear~\cite{ChoiEffros}, where $A$ is nuclear if the identity map $\mathrm{Id}: A\to A$ can be 
approximated pointwise in norm by completely positive finite-rank contractions. 
Many important $C^*$-algebras are nuclear, such as commutative $C^*$-algebras, finite-dimensional 
$C^*$-algebras and $\calK(\calH)$. Furthermore, the class of nuclear $C^*$-algebras is closed 
under quotients, extensions, inductive limits and crossed products by amenable groups, 
see for example \cite[Theorem 15.8.2]{Blackadar}. A more thorough introduction 
to nuclear $C^*$-algebras can be found in the textbook~\cite{BrownOzawa}.

\subsection{Fredholm operators from semi-split extensions} \label{subsec:ext_boundary_map}

As discussed in the previous section, if the short exact sequence 
\[
   0 \to B \otimes \calK \to E \to A \to 0
\]
is semi-split, then Kasparov's Stinespring Theorem gives us a representation 
$\pi: A\to \End_B(\calH_B^{\oplus 2} )$ and a projection $P \in \End_B(\calH_B^{\oplus 2} )$ such that 
$[P, \pi(a)] \in \mathbb{K}_B(\calH_B)$ for all $a \in A$. 
We also remark that if $A$ is non-unital, $\pi$ has a unique extension to a unital $\ast$-homomorphism 
$\pi^\sim: A^\sim \to \End_B(\calH_B^{\oplus 2} )$ with $[ P, \pi^\sim( a + \lambda \one)] = [P, \pi(a) + \lambda \one] \in \mathbb{K}_B(\calH_B^{\oplus 2})$ for 
all $a + \lambda \one \in A^\sim$.
This information can be used to extract a Fredholm 
operator on $\calH_B$ and $K$-theoretic information on the algebra $B$.

\begin{lemma}
Let $A$ be a separable $C^*$-algebra, $B$ a $\sigma$-unital $C^*$-algebra  and $0 \to B \otimes \calK \to E \to A \to 0$ a semi-split extension with $B \otimes \calK$ an 
essential ideal in $E$. 
Then for any invertible element $w \in A^\sim$, $P \pi^\sim(w) P$ is a Fredholm operator on $P\cdot \calH_B^{\oplus 2} \cong \calH_B$.
\end{lemma}
\begin{proof}
Because $P$ has compact commutators with the representation of $A^\sim$, we can easily find an inverse modulo compact operators, 
\[
  P \pi(w)P P \pi(w^{-1}) P = P + P [\pi(w), P] \pi(w^{-1}) P = P + \mathbb{K}_B(\calH_B^{\oplus 2}) = \one_{P \calH_B} + \mathbb{K}_B(\calH_B^{\oplus 2}).
\]
Hence $q( P \pi(w) P) = \tau_P(w)$ is invertible and $P\pi(w) P$ is Fredholm.
\end{proof}

\begin{prop}[{\cite[Section 6--7]{Kasparov80}}] \label{prop:K1_to_Fredholm}
Let $A$ be a separable $C^*$-algebra, $B$ a $\sigma$-unital $C^*$-algebra  and $0 \to B \otimes \calK \to E \to A \to 0$ a semi-split 
extension with $B \otimes \calK$ an 
essential ideal in $E$. Then for any invertible $w \in A^\sim$, the Fredholm operator $P \pi(w) P \in \End_B(\calH_B)$ defines 
a class $[P \pi(w) P] \in KK(\C, B)$ that represents 
the image of $[w] \in K_1(A)$ under the boundary map $K_1(A) \xrightarrow{\delta} K_0(B)\cong KK(\C, B)$.
\end{prop}

We can generalise the above result using  van Daele $K$-theory for $\Z_2$-graded algebras~\cite{vanDaele1, BKR2}.

\begin{prop}[{\cite[Proposition 5.9]{BKR2}}] \label{prop:DK_boundary_to_Fredholm}
Let $A$ and $B$ be $\Z_2$-graded algebras with $A$ separable and $B$ $\sigma$-unital. Suppose that 
$0 \to B \otimes \calK \to E \to A \to 0$ a $\Z_2$-graded semi-split extension with $B \otimes \calK$ an 
essential ideal in $E$. Then for any odd self-adjoint unitary $w \in A^{\sim}$, the Fredholm operator 
$P \pi(w) P \in \End_B(\calH_B)$ defines a class $[P \pi(w) P] \in KK(\C, B)$ that represents 
the image of $[w] \in DK(A)$ under the boundary map $DK(A) \xrightarrow{\delta} DK(B\, \hat\otimes \, \Cl_{1,0})\cong KK(\C, B)$.
\end{prop}


\begin{thebibliography}{99}

\bibitem{AMZ}
A. Alldridge, C. Max and M. R. Zirnbauer. \emph{Bulk-boundary correspondence for disordered free-fermion topological phases}. 
Comm. Math. Phys., \textbf{377} (2020), no. 3, 1761--1821.

\bibitem{AMP}
W. Amrein, M. M\u{a}ntoiu and  R. Purice. 
\emph{Propagation properties for Schr\"{o}dinger operators affiliated with certain $C^*$-algebras}.
Ann. Henri Poincar\'e, \textbf{3} (2002), no.6, 1215--1232.

  \bibitem{AS69}
 {M.~F. Atiyah} and I.~M. Singer. \emph{Index theory for
  skew-adjoint {F}redholm operators}. Inst. Hautes \'{E}tudes Sci. Publ. Math.
  (1969), no.~37, 5--26. 

\bibitem{Blackadar}
B. Blackadar. \emph{$K$-theory for operator algebras}. Second edition. Mathematical Sciences Research Institute Publications, \textbf{5}. 
Cambridge University Press, Cambridge (1998). xx+300 pp.  


\bibitem{B22}
C. Bourne. \emph{Index theory of chiral unitaries and split-step quantum walks}. 
SIGMA Symmetry Integrability Geom. Methods Appl., \textbf{19} (2023), Paper No. 053.

\bibitem{BKR2}
C. Bourne, J. Kellendonk and A. Rennie. \emph{The Cayley transform in complex, real and graded $K$-theory}. 
 Internat. J. Math., \textbf{31} (2020), no. 9, 2050074, 50 pp.

\bibitem{BrownOzawa}
N. P. Brown and N. Ozawa. \emph{$C^*$-algebras and finite-dimensional approximations}. 
Graduate Studies in Mathematics, \textbf{88}. 
American Mathematical Society, Providence, RI, 2008. xvi+509 pp.

\bibitem{ChoiEffros}
M. D. Choi and E. G. Effros. \emph{The completely positive lifting problem for $C\sp*$-algebras}. Ann. of Math. (2), \textbf{104} (1976), no. 3, 585--609. 

\bibitem{vanDaele1} 
A. van Daele. \emph{$K$-theory for graded Banach algebras I}. {Quart. J. Math. Oxford Ser. (2)}, 
\textbf{39} (1988), no. 154, 185--199.

\bibitem{Iwatsuka}
G. De Nittis, J. Gomez and D. Polo Ojito. \emph{On the K-theory of magnetic algebras: Iwatsuka case}. arXiv:2410.04531 (2024).

\bibitem{DS85}
V. M. Deundyak and B. Ya. Shteĭnberg. \emph{The index of the convolution operators with slowly changing coefficients on abelian groups}. 
Funktsional. Anal. i Prilozhen., \textbf{19} (1985), no. 4, 84--85.

\bibitem{DZ24}
A. Drouot and X. Zhu. \emph{Topological edge spectrum along curved interfaces}. Int. Math. Res. Not. IMRN, (2024), no. 22, 13870--13889.

\bibitem{vdDungen17}
K. van den Dungen. \emph{The index of generalised Dirac--Schr\"{o}dinger operators}. 
J. Spectr. Theory, \textbf{9} (2019), no. 4, 1459--1506.


\bibitem{Georgescu1}
V. Georgescu and A. Iftimovici. \emph{Crossed products of $C^\ast$-algebras and spectral analysis of quantum Hamiltonians}. 
Comm. Math. Phys., \textbf{228} (2002), no. 3, 519--560.

\bibitem{Georgescu2}
  V. Georgescu and A. Iftimovici. \emph{Localizations at infinity and essential spectrum of quantum Hamiltonians. I. General theory}. 
  Rev. Math. Phys., \textbf{18} (2006), no. 4, 417--483.


\bibitem{Hayashi23}
S. Hayashi. \emph{An index theorem for quarter-plane Toeplitz operators via extended symbols and gapped invariants related to corner states}. Comm. Math. Phys.,  
\textbf{400} (2023), no. 1, 429--462.

\bibitem{KasparovStinespring}
G. G. Kasparov. \emph{Hilbert {$C^{\ast}$}-modules: theorems of {S}tinespring and {V}oiculescu}. 
J. Operator Theory, \textbf{4} (1980), no. 2, 133--150.

\bibitem{Kasparov80}
  G. G. Kasparov. \emph{The operator $K$-functor and extensions of $C^*$-algebras}. 
  {Math. USSR Izv.}, \textbf{16} (1981), 513--572.

\bibitem{Kellendonk15}
J. Kellendonk. \emph{On the $C^*$-algebraic approach to topological phases for insulators}. 
Ann. Henri Poincar\'{e}, \textbf{18} (2017), no. 7, 2251--2300.

\bibitem{KSBV}
M. Kotani, H. Schulz-Baldes and C. Villegas-Blas. \emph{Quantization of interface currents}. J. Math. Phys., \textbf{55} (2014), no. 12, 121901, 9 pp.

\bibitem{Kubota16}
Y. Kubota. \emph{Notes on twisted equivariant K-theory for $\rm C^*$-algebras}. 
Internat. J. Math., \textbf{27} (2016), no. 6, 1650058, 28 pp.

\bibitem{Kubota21}
Y. Kubota. \emph{The bulk-dislocation correspondence for weak topological insulators on screw-dislocated lattices}. 
J. Phys. A, \textbf{54} (2021), no. 36, Paper No. 364001, 18 pp.

\bibitem{Lance} 
E. C. Lance. \emph{Hilbert $C^*$-Modules: A Toolkit for Operator Algebraists}.
Cambridge University Press, Cambridge (1995).

\bibitem{LT22}
M. Ludewig and G. T.   Thiang. \emph{Cobordism invariance of topological edge-following states}. Adv. Theor. Math. Phys., \textbf{26} (2022), no. 3, 673--710. 


\bibitem{Mantoiu02}
M. M$\breve{\text{a}}$ntoiu. \emph{$C^\ast$-algebras, dynamical systems at infinity and the essential spectrum of generalized Schr\"{o}dinger operators}. 
J. Reine Angew. Math., \textbf{550} (2002), 211--229.

\bibitem{MantoiuGpoid}
M. M\u{a}ntoiu. \emph{$C^*$-algebraic spectral sets, twisted groupoids and operators}. arXiv:1809.03347 (2018).

\bibitem{MPR}
M. M\u antoiu, R. Purice and S. Richard.
\emph{Spectral and propagation results for magnetic  Schr\"{o}dinger operators; a $C^*$-algebraic framework}. J. Funct. Anal., \textbf{250} (2007), no.1, 42--67.

\bibitem{Pedersen}
G. K. Pedersen. \emph{$C^*$-algebras and their automorphism groups}. Second edition. Pure and Applied Mathematics (Amsterdam). Academic Press, London (2018). xviii+520 pp.

\bibitem{ProdanHigherOrder}
D. Polo Ojito, E. Prodan and T. Stoiber. \emph{$C^*$-framework for higher-order bulk-boundary correspondences}. 
Comm. Math. Phys., \textbf{406} (2025), no. 10, Paper No. 233.

\bibitem{PPSadiabatic}
D. Polo Ojito, E. Prodan and T. Stoiber. \emph{A space-adiabatic approach for bulk-defect correspondences in lattice models of topological insulators}.
arXiv:2410.24097 (2024).


\bibitem{ProdanKK}
E. Prodan. \emph{Topological lattice defects by groupoid methods and Kasparov's $KK$-theory}. J. Phys. A, \textbf{54} (2021), no. 42, Paper No. 424001, 18 pp. 

\bibitem{PSBbook}
E. Prodan and H. Schulz-Baldes. \emph{Bulk and Boundary Invariants for Complex Topological Insulators: From $K$-Theory to Physics}. 
Springer, Berlin (2016).

\bibitem{RSS98}
V. S. Rabinovich, S. Roch and B. Silbermann. \emph{Fredholm theory and finite section method for band-dominated operators}.
Integral Equations Operator Theory, \textbf{30} (1998), no. 4, 452--495.

\bibitem{Ri05}
S. Richard. 
\emph{Spectral and scattering theory for  Schr\"{o}dinger operators with Cartesian anisotropy}
Publ. Res. Inst. Math. Sci., \textbf{41} (2005) no. 1, 73--111.




\bibitem{Roe03}
J. Roe. \emph{Lectures on coarse geometry}. University Lecture Series, \textbf{31}. 
American Mathematical Society, Providence, RI, 2003. viii+175 pp.

\bibitem{Roe04}
J. Roe. \emph{Paschke duality for real and graded $C^*$-algebras}. Q. J. Math., \textbf{55} (2004), no. 3, 325--331.

\bibitem{Sunada}
T. Sunada. \emph{Topological crystallography. With a view towards discrete geometric analysis}. 
Surveys and Tutorials in the Applied Mathematical Sciences, 6. Springer, Tokyo (2013). xii+229 pp.

\bibitem{ThiangEdge}
G. C. Thiang. \emph{Edge-following topological states}. J. Geom. Phys., \textbf{156} (2020), 103796, 17 pp.

\bibitem{Wahl07}
C. Wahl. \emph{On the noncommutative spectral flow}. J. Ramanujan Math. Soc., \textbf{22} (2007), no. 2, 135--187.



\bibitem{WO}
N. E. Wegge-Olsen. \emph{K-theory and $C^*$-algebras, a friendly approach}. Oxford Science Publications. 
The Clarendon Press, Oxford University Press, New York (1993).  xii+370 pp.

\bibitem{Wil07}
D. P. Williams. \emph{Crossed products of $C^*$-algebras}.
American Mathematical Society, Providence (2007).

\end{thebibliography}
\end{document}